\ifdefined\bioformat
  \documentclass{bioinfo}
\else
  \documentclass{bioinfo_arxiv}
\fi

\usepackage{xspace} % Correct macro spacing
\usepackage{multibib}
\newcites{appendix}{References}
\def\ourbm{\textsc{SMaSH}\xspace}
\def\rescue{\textsc{Rescue}\xspace}
\def\rescuega{\textsc{Rescue-GA}\xspace}
\def\ourwebsite{\url{smash.cs.berkeley.edu}}
\def\ourrepo{\url{github.com/amplab/smash}}
\usepackage{tikz}
\usetikzlibrary{shapes,backgrounds}
\usepackage{url}
\usepackage{amsmath}
\usepackage{multirow}
\usepackage{amsthm,amssymb}
\usepackage{pdflscape}
\newtheorem{prop}{Proposition}
\newtheorem*{prop*}{Proposition}

\usepackage{algorithm}
\usepackage{algorithmic}

\def\firstcircle{(0,0) circle (1.5cm)}
\def\secondcircle{(60:2cm) circle (1.5cm)}
\def\thirdcircle{(0:2cm) circle (1.5cm)}
\newcommand{\whiteintersection}{
    \begin{scope}
      \clip \firstcircle;
      \clip \secondcircle;
      \fill[white] \thirdcircle;
    \end{scope}
}
\newcommand{\comment}[1]{}

\newcommand{\TP}{\mathrm{TP}}
\newcommand{\TN}{\mathrm{TN}}
\newcommand{\FP}{\mathrm{FP}}
\newcommand{\FN}{\mathrm{FN}}
\newcommand{\present}{\mathrm{present}}
\newcommand{\absent}{\mathrm{absent}}
\newcommand{\precision}{\mathrm{precision}}
\newcommand{\recall}{\mathrm{recall}}
\newcommand{\tooslow}{$>400h$}
\newcommand{\toomuch}{$>\$1000$}

\usepackage{pdflscape}

\copyrightyear{2013}
\pubyear{2013}

\newcommand{\markup}[1]{\textcolor{black}{#1}}
\begin{document}

\firstpage{1}

\title[\ourbm]{\ourbm: A Benchmarking Toolkit for \markup{Human Genome} Variant Calling}
\author[Talwalkar \textit{et~al}]{Ameet Talwalkar\,$^{1,\dagger, *}$, Jesse
Liptrap\,$^{1,\dagger}$, Julie Newcomb\,$^{1}$,
Christopher Hartl\,$^{1,3}$, Jonathan
Terhorst\,$^{2}$, Kristal Curtis\,$^{1}$, Ma'ayan Bresler\,$^{1}$, Yun
S. Song\,$^{1,2}$, Michael I. Jordan\,$^{1,2}$,  and David
Patterson\,$^{1,*}$}
\address{$^{1}$Department of Electrical Engineering and Computer Science,
UC Berkeley, Berkeley, CA. \\ 
$^{2}$Department of Statistics, UC Berkeley, Berkeley, CA. \\
$^{3}$The Broad Institute of Harvard and MIT, Cambridge, MA. \\
$^{\dagger}$These authors contributed equally.\\
$^{*}$To whom correspondence should be addressed.
}

\ifdefined\bioformat
  \history{Received on XXXXX; revised on XXXXX; accepted on XXXXX}
  \editor{Associate Editor: XXXXXXX}
\else
  \history{}
  \editor{}
\fi

\maketitle
\newcommand{\absmotivation}{\section{Motivation:}}
\newcommand{\absresults}{\section{Results:}}
\newcommand{\absavailable}{\section{Availability:}}
\newcommand{\abscontact}{\section{Contact:} \href{}{ameet@cs.berkeley.edu, pattrsn@cs.berkeley.edu}}

\begin{abstract}

%\section{Motivation:} Next-generation sequencing holds the promise for
\absmotivation 
%Although next-generation sequencing technology holds the promise for great advances
%in science and medicine, 
Computational methods are essential to extract actionable information from raw
sequencing data, and to thus fulfill the promise of next-generation sequencing
technology. Unfortunately, computational tools developed to call variants from
human sequencing data
%are difficult to use
disagree on many of their predictions, and current methods to evaluate
accuracy and computational performance are ad-hoc and incomplete. Agreement on
benchmarking variant calling methods would stimulate development of genomic
processing tools and facilitate communication among researchers.

%\section{Results:} We propose such a benchmarking methodology called the
\absresults We propose \ourbm, a benchmarking methodology for evaluating \markup{human
genome} variant calling algorithms. We generate synthetic datasets, organize and
interpret a wide range of existing benchmarking data for real genomes, and
propose a set of accuracy and computational performance metrics for evaluating
variant calling methods on this benchmarking data.  Moreover, we illustrate the
utility of \ourbm to evaluate the performance of some leading single nucleotide
polymorphism (SNP), indel, and structural variant calling algorithms. 

%\section{Availability:} We provide free and open access online to \ourbm, along
\absavailable We provide free and open access online to the \ourbm toolkit,
along with detailed documentation, at \ourwebsite.  

%We ultimately envision building a platform for researchers to share additional
%benchmarking data with the community. 

%\section{Contact:} \href{}{ameet@cs.berkeley.edu, pattrsn@cs.berkeley.edu}
\abscontact
\end{abstract}

\section{Introduction}
Next-generation sequencing is revolutionizing biological and clinical research.
%\citep{Schuster08}.
%The rapid growth of human genome sequencing technology presents an exciting
%Big Data challenge. 
Long hampered by the difficulty and expense of obtaining genomic data, life
scientists now face the opposite problem: faster, cheaper technologies are
beginning to generate massive amounts of new sequencing data that are
overwhelming our technological capacity to conduct genomic analyses
%\citep{Mardis10,McPherson09}.  Computational processing will soon become the
\citep{Mardis10}.  Computational processing will soon become the
bottleneck in genome sequencing research, and as a result, computational
biologists are actively developing new tools to more efficiently and accurately
process human genomes and call variants, e.g., SAMTools \citep{Li09}, GATK
\citep{DePristo11}, Platypus \citep{Platypus}, %CASAVA \citep{Casava},
BreakDancer \citep{Chen09}, Pindel \citep{Ye09}, Dindel \citep{Albers11}, and so on. 
%Telescoper \citep{Bresler12}.  

Unfortunately, SNP callers disagree as much as
%20\% of the time \citep{Haussler12,Lyon12} and there is even less consensus in
20\% of the time \citep{Lyon12} and there is even less consensus in the outputs
of structural variant algorithms \citep{Alkan11}. Moreover, reproducibility,
interpretability, and ease of setup and use of existing software are pressing
issues currently hindering clinical adoption~\citep{Nekrutenko12}.  Indeed,
reliable benchmarks are required in order to measure accuracy, computational
performance, and software robustness, and thereby improve them.

In an ideal world, benchmarking data to evaluate variant calling algorithms
would consist of several fully sequenced, perfectly-known human genomes.  However,
ideal validation data do not exist in practice.  Technical limitations, such
as the difficulty in accurately sequencing low-complexity regions, along with
budget constraints, such as the cost to generate high-coverage Sanger reads,
limit the quality and scope of validation data.  Nonetheless, significant
resources have already been devoted to generate subsets of benchmarking data
that are substantial enough to drive algorithmic innovation.  Alas, the
existing data are not curated, thus making it extremely difficult to access,
interpret, and ultimately use for benchmarking purposes.

Due to the lack of curated ground truth data, current benchmarking efforts with
sequenced human genomes are lacking. The majority of benchmarking today relies
on either simulated data or a limited set of validation data associated with
real-world datasets. Simulated data are valuable but do not tell the full
story, as variant calling is often substantially easier using synthetic reads
generated via simple generative models. Sampled data, as mentioned above, are
not well curated, resulting in benchmarking efforts (such as the Genome in a
Bottle Consortium (GBC) \citep{Zook11} and the Comparison and Analytic Testing
resource (GCAT) \citep{GCAT}) that rely on a single dataset with a limited
quantity of validation data.

%In addition to issues related to the underlying benchmarking data, 
Rigorously evaluating predictions against a validation dataset presents
several additional challenges.  Consensus-based evaluation approaches, employed
in various benchmarking efforts, e.g., \cite{1kGenomes10, Kedes11, DePristo11, GCAT},
may be misleading.  Indeed, different methods may in fact make similar errors, a fact which
remains hidden without ground truth data.  In cases where `noisy' ground truth
data are used, e.g., calls based on Sanger-sequencing with some known error rate
or using SNP chips with known error rates, accuracy metrics should account for the
effect of this noise on predictive accuracy.  Additionally, given the
inherent ambiguity in the VCF format used to represent variants, evaluation can
be quite sensitive to the (potentially inconsistent) representations of
predicted and ground truth variants. 
Moreover, due to the growing need to efficiently process raw sequencing data,
computational performance is an increasingly important yet to date largely
overlooked factor in benchmarking.  
There currently exist no benchmarking methodologies that -- in a consistent and
principled fashion -- account for noise in validation data, ambiguity in variant
representation, or computational efficiency of variant calling methods.
%Moreover, there is a growing need to efficiently process large amounts of raw
%sequencing data, thus rendering computational performance an increasingly
%important, yet to date largely overlooked, factor in benchmarking. 
%To the best of our knowledge, there exist no benchmarking methodologies that,
%in a consistent and principled fashion, account for noise in validation data,
%ambiguity in variant representation, or computational efficiency of variant
%calling methods.  
%e.g., generating synthetic data
%using different simulators or using real sequence data corresponding to
%different genomes and/or different chromosomes, 

Without any standard datasets and evaluation methodologies, research groups
inevitably perform ad-hoc benchmarking studies, working with different datasets
and accuracy metrics, and performing studies on a variety of computational
infrastructures.  Competition-based exercises, e.g., \cite{Kedes11, assemblathon}, are a
popular route for benchmarking that aim to address some of these
inconsistencies, but they are ephemeral by design and often suffer from the
same data and evaluation pitfalls described above. 

In short, the lack of consistency in datasets, computational frameworks, and
evaluation metrics across the field prevents simple comparisons across
methodologies, and in this work we make a first attempt at addressing these
issues. We propose \ourbm,  a standard methodology for benchmarking variant
calling algorithms based on a suite of \textbf{S}ynthetic, \textbf{M}ouse,
\textbf{a}nd \textbf{S}ampled \textbf{H}uman data. \ourbm leverages a rich set
of validation resources, in part bootstrapped from the patchwork of existing
data. %, and addresses the aforementioned issues with current benchmarking
%approaches. 
We provide free and open access to \ourbm, which consists of:
\begin{itemize}
\item A set of 5 full-genomes with associated deep coverage short-read datasets
(real and synthetic);
\item 3 contaminated variants of these datasets that mimic real-world use
cases~\citep{DePristo13} and test the robustness of variant callers in terms
of accuracy and required computational resources; 
\item Ground truth validation data for each genome along with detailed error
profiles;
\item Accuracy metrics that account for the uncertainty in validation data;
\item Methodology to resolve the ambiguity in variant representations, resulting
in stable measurements of accuracy; and
\item Performance metrics to measure computational efficiency (and implicitly
measure software robustness) that leverage the Amazon Web Services cloud
computing environment.
\end{itemize}

%We believe that \ourbm is an important stepping stone 
\ourbm is designed to facilitate progress in algorithm development by making it
easier for researchers to evaluate their systems against each other.  

%we think that the resources provided by \ourbm, in part bootstrapped from
%the patchwork of existing data, will stimulate algorithmic innovation

%In summary, the computational task of extracting variants from short-read data
%is too important to wait for an `ideal' benchmarking dataset to be generated,
%and the resources provided by \ourbm, in part bootstrapped from the patchwork
%of existing data, is sufficient to drive innovation.  

%The computational task of extracting variants from short-read data is too
%important to wait for an `ideal' benchmarking dataset to be generated, and
%\ourbm, in part bootstrapped from the patchwork of existing data, is sufficient
%to drive innovation in the short and medium term. 

%Nonetheless, it is important to note that
%our work is a first attempt at compiling benchmarking data for variant calling
%algorithms.  We do not claim that \ourbm is comprehensive or unbiased, and a
%long-term solution to benchmarking involves devoting significant resources to
%generate comprehensive datasets, e.g., via high-coverage Sanger sequencing as
%described above or via Illumina's proposed `platinum' genomes involving deep
%sequencing (e.g., 200x coverage) of families across generations and leveraging
%the principle of Mendelian inheritance. 

\begin{figure}[!t]
\begin{center}
\begin{tabular} {c|ccc}
  \begin{tikzpicture}[scale=.3]
    \draw \firstcircle node[below] {$\mathbf{R}$};
    \draw \secondcircle node [above] {$\mathbf{C}$};
    \draw \thirdcircle node [below] {$\mathbf{H}$};
    % Next, we want the highlight the intersection of all three circles:
    \begin{scope}
      \clip \firstcircle;
      \clip \secondcircle;
      \fill[green] \thirdcircle;
    \end{scope}
  \end{tikzpicture} \quad&
\quad
\begin{tikzpicture}[scale=.3]
    \draw \firstcircle node[below] {$\text R$};
    \draw \secondcircle node [above] {$\mathbf C$};
    \draw \thirdcircle node [below] {$\mathbf H$};
    \begin{scope}
      \clip \secondcircle;
      \fill[blue] \thirdcircle;
    \end{scope}
    \whiteintersection{}
  \end{tikzpicture} \quad &
  \begin{tikzpicture}[scale=.3]
    \draw \firstcircle node[below] {$\mathbf R$};
    \draw \secondcircle node [above] {$\mathbf C$};
    \draw \thirdcircle node [below] {$\text H$};
    \begin{scope}
      \clip \firstcircle;
      \fill[red] \secondcircle;
    \end{scope}
    \whiteintersection
  \end{tikzpicture} \quad &
  \begin{tikzpicture}[scale=.3]
    \draw \firstcircle node[below] {$\mathbf R$};
    \draw \secondcircle node [above] {$\text C$};
    \draw \thirdcircle node [below] {$\mathbf H$};
    \begin{scope}
      \clip \thirdcircle;
      \fill[yellow] \firstcircle;
    \end{scope}
    \whiteintersection
  \end{tikzpicture} \\
{\footnotesize Ideal} \quad & \quad {\footnotesize Synthetic} \quad& {\footnotesize Mouse} \quad& {\footnotesize Sampled Human} \\
\end{tabular}
\end{center}
\caption{An `ideal' benchmarking dataset satisfies three properties: it
contains real reads ($\mathbf{R}$), it includes comprehensive validation of the
underlying genome ($\mathbf{C}$), and its underlying genome is human
($\mathbf{H}$). \ourbm contains three types of benchmarking datasets, each of
which satisfies two of the three desirable properties of an ideal dataset, so
as to cover all three properties.}
\label{fig:actual_bm}
\end{figure}
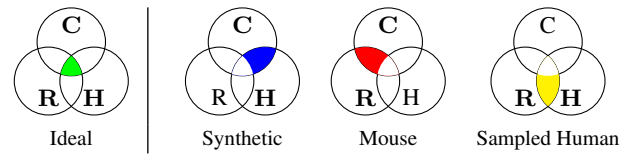

\section{Methods}
\label{sec:approach}
\subsection{Benchmarking Datasets}
\label{ssec:data}
%We focus on benchmarking algorithms that perform human genome reconstruction
%and variant calling using short-read data.  
In this section, we describe the benchmarking datasets contained within
\ourbm. 
%and then present the performance metrics we use to evaluate the
%performance of alignment, SNP-calling and structural variant calling
%algorithms. 
A `benchmarking dataset' consists of three components.  The first
two components are the inputs to the variant calling algorithm to be
benchmarked, namely short reads generated from next-generation sequencing
technology and a reference genome (used for alignment and variant
representation).  The third component is the validation data (represented via
the standard VCF format) that are used to evaluate the quality of an algorithm's
predictions. 

The left panel of Figure~\ref{fig:actual_bm} illustrates the three desired properties of a
benchmarking dataset. Ideally, we would like to evaluate variant calling
performance on a human genome ({\textbf H}), have access to  comprehensive
validation ({\textbf C})  of the underlying sequenced genome, and call variants
using real reads ({\textbf R}), that is, reads generated by an actual sequencing
machine and not a simulator. To the best of our knowledge, no existing dataset
satisfies all three properties.  
Instead, \ourbm consists of three types of benchmarking datasets 
% (`Synthetic', `Mouse,' and 'Sampled Human')
that satisfy two of these three properties, as depicted in the three
panels on the right of Figure~\ref{fig:actual_bm}.  Additionally, for each type of
dataset, we also include a contaminated version in which the short reads are
contaminated with reads from a separate genome, mimicking the impurities that
can be introduced in practice while preparing a sample and/or using a
contaminated sequencing machine, and thus testing the robustness of variant
callers in this challenging and realistic setting.
Table~\ref{tab:data_summary} summarizes our validation data, and we next
provide details about these datasets.

\begin{table*}[ht!]
\caption{Summary of \ourbm's validation datasets.}
\centering
  \label{tab:data_summary}
  \small{ 
  \begin{tabular}[c]{cccccccc}
  \bf{Type} & \bf{Genome} & \bf{Validation Error} & \bf{Sequencer} &
  \bf{Length (bp)}& \bf{Insert Size (bp)} & \bf{Coverage} \\
    \hline
    \multirow{2}{*}{Synthetic} & Venter & \multirow{2}{*}{none}  &
    \multirow{2}{*}{simNGS} & \multirow{2}{*}{101} & \multirow{2}{*}{400} & \multirow{2}{*}{30x} \\
    & Contam. Venter & & & &  \\
    \hline
    \multirow{2}{*}{Mouse} & B6 strain & \multirow{2}{*}{$0.2\%$ (SNP/Indel), $0.3\%$ (SV)} &
    \multirow{2}{*}{GAIIx} & \multirow{2}{*}{101} & \multirow{2}{*}{-34} &
    \multirow{2}{*}{58.6x} \\
    & Contam. B6 strain & & & & &  \\ 
  \hline
    \multirow{4}{*}{Human} & NA12878 & \multirow{4}{*}{$0.04\%$ (SNP), $1\%$
    (SV)} & HiSeq2000 & 101 & 300 & 50x \\
     & Contam. NA12878 & &  HiSeq2000 & 101 & 300 & 50x \\
     & NA18507 & &  HiSeq2500 & 100 & 300 & 44x \\
     & NA19240 & & GAIIx & 101 & 296 & 49x \\
    \hline
  \end{tabular}
 }
\end{table*}

\subsubsection{Synthetic Datasets}

We derive our synthetic datasets from J. Craig Venter's genome (HuRef).  \markup{We
create a diploid sample genome by applying the HuRef variants provided by
\cite{Levy07} to the hg19 reference genome.}\footnote{\markup{We create an unphased
diploid sample genome by starting with two copies of hg19 and inserting each
HuRef variant into one or both of these copies depending on its zygosity.}}
 %\citep{IHGSC2004}.
We simulate Illumina reads from the sample genome using simNGS \citep{simNGS}
with its default settings.
%Simulated reads are 101bp paired-end reads with insert size
%400 and 30x coverage.  
Our second dataset uses the same validation data as the
first, along with a version of Venter's short-read data contaminated by similar
short-read data derived from an approximation of James Watson's genome.

\paragraph{\textbf{Validation Error Profile:}} \markup{Any errors in HuRef will be
carried through to our synthetic genome, and it is likely that errors in HuRef
are sequence-context specific. Nonetheless, since the sample genome
and the reads are synthetically generated, the VCF files contain noiseless
ground truth data.} 

\subsubsection{Mouse Datasets}\label{subsec:mouse}
The mouse datasets leverage existing mouse genomic data associated with the
canonical mouse reference as well as from the Mouse Genomes Project (MGP)
%(\citep{Keane11}, \citep{Yalcin11}, \citep{Yalcin12}). From this data we create
\citep{Yalcin11}. From this data, we create benchmarking datasets with real
reads and with comprehensive validation, using the canonical homozygous mouse
reference as our sample genome \citep{Church09}. Our first
dataset consists of paired-end reads from the B6 mouse strain \citep{Gnerre11},
a VCF derived from differences between the mouse reference (based on the B6
mouse strain), and a `fake' reference we created using an alternative mouse
strain \markup{(the DBA mouse strain)}. Figure~\ref{fig:mouse} illustrates the process
by which we create this dataset, and further details are provided in
%Supplementary Material~\ref{app:mouse_prep}.
Supplementary Material~B.2.
Our second dataset uses the same validation data as the first, along with a
version of the B6 short reads contaminated by short reads corresponding to a
human genome (NA12878). \markup{We use these human reads because, to the best of our
knowledge, they are the only publicly available reads generated by the same
sequencing methodology as the mouse reads~\citep{Gnerre11}.}
%, which involves \tvsim (but without its
%read-generating functionality). 
%We describe below the details of this process.

\paragraph{\textbf{Validation Error Profile:}}

There are two main sources of error in this dataset, namely  errors in the mouse
reference genome itself and genetic differences between the mouse reference
genome and the individual from which short reads were produced.  Based on
calculations detailed in 
%Supplementary Material~\ref{app:mouse_prep}, 
Supplementary Material~B.2,
we upper bound the
error rates for SNPs and indels at $0.2\%$ and the error rate for SVs at
$0.3\%$. Finally, it is worth noting that there are systematic differences
between mouse and human genomes.  Mouse segmental duplication is more
intrachromosomal, and human intrachromosomal duplication is more high-identity
\citep{Church09}. As a result, variant calling performance may vary 
between the mouse and human datasets.

\begin{figure}
\begin{center}
\includegraphics[scale=.39]{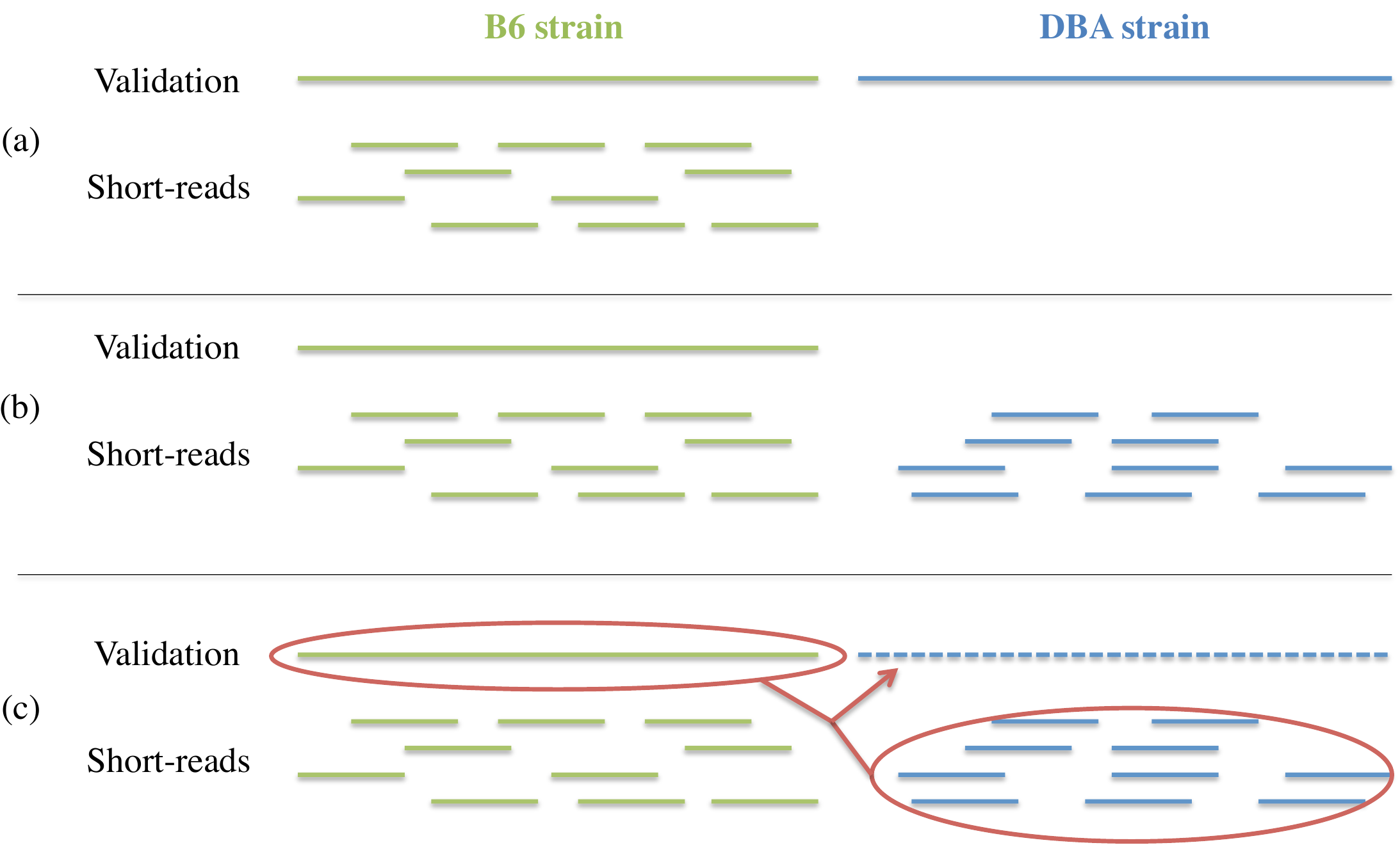}
\end{center}
\caption{Schematic illustrating process by which \ourbm's first `Mouse' dataset
is generated. (a) Our ideal set up in which the B6 strain (with comprehensive
validation and corresponding short reads) serves as the sample and the DBA
strain serves as the reference. (b) Publicly available data (note that the B6
validation data are the canonical mouse reference). (c) Construction of an
approximate DBA validation set (the `fake' reference)  by leveraging a rough
set of variants for the DBA strain called relative to the canonical reference.}
\label{fig:mouse}
\end{figure}

\subsubsection{Sampled Human Datasets}
\label{ssec:sampled_human}
Our real human genomes consist of three well-studied human genomes, including a
European female (NA12878), a Nigerian male (NA18507), and a Nigerian female
(NA19240). Our validation data consist of subsets of validated SNP and SV
information. \ourbm also includes high-coverage Illumina reads for each of
these datasets. 
%available for three of the nine individuals for whom we have validated SVs,
%while low-coverage reads are available for all 270 HapMap2 individuals.
%Details about our validated variants and the short-read data are described
%below:
%particular, we compile structural variants (currently insertions and deletions)
%for nine individuals; SNPs for these same individuals as well as the 261 other
%HapMap2 individuals; and more reliable validation data on a subset of those
%SNPs for eight of the nine individuals.  High-coverage Illumina reads are
%available for three of the nine individuals for whom we have validated SVs,
%while low-coverage reads are available for all 270 HapMap2 individuals.
%Details about our validated variants and the short-read data are described
%below:

We derive our validated SNPs from the intersection of calls from two SNP chips
from HapMap2: Perlegen and Illumina BeadArray \citep{Frazer07}. We chose these
chips due to the substantial intersection of their call sites, and because
their calls could be readily disambiguated, unlike the two chip technologies
used in HapMap3 \citep{Consortium10}. The intersection of their results in
NA12878, yields 132K calls, of which 55K are non-reference.
%\item For extra reliability, we provide a smaller set of SNPs by taking another
%intersection, with SNPs called from fosmid end alignments \citep{Kidd08}, which
%are only available for eight individuals, i.e., eight of the nine individuals
%for whom we have SVs (except for NA15510).  For NA12878, there are 0.6 million
%fosmid SNPs, whose intersection with the 55K SNPs found previously yields 9.6K
%triple-redundant non-reference SNPs.
Our SV validation data is 169 insertions and deletions called from
alignment of finished fosmid sequence \citep{Kidd10insertions, Kidd10insights}.
%\footnote{Validated inversion data was
%also available, but we chose not to use this data because sequence for these
%inversions is only fully known at one breakpoint: there is a chance that the
%ther breakpoint might contain extra sequence \citep{Kidd10insights}.}

% \todo{Julie: can you update the descriptions of SNPs and SVs as necessary. in
% particular do you know where the inversion data comes from?  is it also from
% \citep{Kidd10insights} or another source?  Yes, it's from the same source,
% as cited in the footnote.}

Finally, we include a contaminated version of our
NA12878 dataset, in which the NA12878 short reads are contaminated by short reads
corresponding to this individual's husband (NA12877),
generated by the same sequencing methodology by Illumina's Platinum Genomes project.

%Moreover, reads of coverage between 2x and 6x are available \citep{Consortium12}
%for all HapMap2 individuals, from Illumina, SOLID, or 454 platforms.

%were submitted by Illumina to the Sequence Read Archive
%\citep{Leinonen10} as accession ERX009609; they have length 100bp, insert size
%300bp, and at least 30x coverage.  The NA19240 reads were submitted by the
%Genetic Testing Reference Materials Coordination Program to the SRA as
%accession SRX152746.  

\paragraph{\textbf{Validation Error Profile:}}

The error in our validated SNP data is due to errors in the underlying SNP chip
technologies used to generate the data. Moreover, there are various sources of
error for our validated SVs, associated with generating and processing the
%fosmid sequences.  As detailed in Supplementary Material~\ref{app:human_prep}, we upper bound
fosmid sequences.  As detailed in Supplementary Material~B.3, we upper bound
the error rate for SNPs at $0.04\%$, and the error rate for SVs at $1.0\%$. 

\subsection{Evaluation Metrics}
\label{ssec:metrics}

We now discuss our set of evaluation metrics of variant calling algorithms
against the benchmarking datasets described in
Section~\ref{ssec:data}.
%\footnote{Our Python evaluation script, along with
%all benchmarking datasets, can be downloaded at \ourwebsite.} 
We propose the use of both accuracy and computational performance metrics.
Moreover, although \ourbm focuses on reference-based variant calling methods,
it \emph{does not include alignment-specific metrics}, as alignment is an
intermediate step (albeit an important one) in the process of variant calling.
We believe that improvements in alignment should be measured as a function of
their impact on variant calling, both in terms of accuracy and computational
performance.

\noindent\textbf{Accuracy}: We report two standard metrics from
information retrieval. The first metric, \emph{recall}, measures the
`probability of calling a validated variant,' while the second metric,
\emph{precision}, measures the `probability that a called variant is
%correct.'\footnote{See Section~\ref{sec:rec_prec} in the Supplementary Material
correct.'\footnote{See Section~C in the Supplementary Material
for a more detailed discussion of the use of recall and precision in the
context of variant calling.}
%of negative validation data.}  
%Indeed, the behavior of a variant caller is
%determined by four quantities: each call is positive or negative (depending
%solely on the caller) and true or false (depending also the validation
%dataset), so we count true (TP) and false positives (FP)  and true (TN) and
%false negatives (FN). One might be tempted to simply report all four numbers
%above, however, as stated in Proposition~\ref{prop:coordinate_change}
%the space of
%caller behaviors can be parametrized  two-dimensional, and to
%parametrize the remaining two degrees of freedom we advocate the convenient and
%very widespread notions of recall and precision:
%\
%However where the total numbers of
%variant presences ($\present = \TP + \FN$) and absences ($\absent = \FP + \TN$)
%are intrinsic to the dataset, i.e., independent of the algorithm.
%The space of
%caller behaviors on a particular dataset is therefore two-dimensional, and to
%parametrize the remaining two degrees of freedom we advocate the convenient and
%very widespread notions of recall and precision:
For SNPs and indels (50bp or less) we measure alternate alleles and exact breakpoints, 
thus checking zygosity but ignoring phasing.
% For structural variants,
% we check for correct breakpoint detection within some tolerance of acceptable
% error, with the tolerance a function of the acceptable discrepancy in length
% and overlap of a prediction relative to the true structural variant
\markup{For structural variants, we ignore zygosity, and evaluate left breakpoint and
length, both approximately and exactly.  In the former case, we use an error
tolerance of 100bp, which is within a read-length of the true event and thus
sufficiently close to allow for alternative methods, such as targeted assembly,
to reconstruct the event.  We also present the exact evaluation to highlight
variations in accuracy.} In some situations, such as our human SNP or SV
validation data, the validation data have positive labels but little or no
negative labels, and in these situations
%i.e., the
%validation data do not say that there is \emph{not} a SV at any particular
%location.  Thus, 
only recall (and not precision) is reported.  Additionally, our computation of
recall does not explicitly take into account the impact of sampling in the
context of \ourbm's sampled human benchmarking datasets.

\noindent\textbf{Computational Performance}: We use two chief metrics to
measure computational performance, namely \emph{hours per genome} and
\emph{dollars per genome}, and we benchmark performance on Amazon Web
Services (AWS). AWS's cloud infrastructure allows for reproducible benchmarking
and ensures robust implementations of variant calling algorithms. When using \ourbm, 
researchers can benchmark algorithms on AWS using their
preferred compute instances, such as single core, multicore, GPU or distributed
cluster, and AWS's pricing mechanism naturally dictates the tradeoff between
cost and time. \markup{We additionally report more fine-grained performance metrics to help
researchers optimize their AWS configuration, including clock time, CPU time, the
maximum number of threads used, the maximum disk space required, and the maximum and average
amount of memory used during the run of the algorithm.}

%More generally, we want some measure of speed / cost for all algorithms
%(perhaps by running everything on AWS).

\subsubsection{Accounting for noisy validation data}
The performance of an algorithm can only be quantified up to the level of noise
in the validation data itself. Since we are working with noisy validation data,
it is crucial to capture this uncertainty when reporting results.  
%Moreover,
%the proper method for incorporating this uncertainty depends on the type of
%validation data.
%
% For synthetic data, we advocate simply reporting precision and recall.  For
% real data, however, we must consider that the ground truth may be incomplete or
% imperfect.
%
To do so, we assume that we have an estimate for the number $E$ of
validation errors.  In settings where we only have positive labels, this
estimate captures the number of positive labels that are incorrectly genotyped
(either a positive label where there should be none, or a validated variant
that is correctly located but incorrectly genotyped).  In settings where we
have access to both positive and negative labels, this estimate measures either
omissions or errors of the type described above. 
%
%As an extreme
%example, for the human SNP dataset, the computation of GATK's recall doesn't
%take into account that HapMap SNPs, including all SNPs in the validation set,
%are baked into GATK's prior.  As a less extreme example, for Mullikin SVs there
%are usually <100 validated SVs for any given individual, and we don't quantify
%the error due to estimating recall from a relatively small sample of true
%variants.
% recall here is only relative to the validated subset of the
%true variants; this limitation is more important than error bounds.
%In other situations, we have access to both positive and negative labels, e.g.,
%in the comprehensive validation setting (mouse, synthetic) or in the case of
%our human SNP validation data.  In these situations, we can compute both recall
%and precision (in the case of human SNPs, the computation is over the sampled
%validation data).  In these situations, there are still some number $E$ of
%validation errors, We now present best-case upper and worst-case lower bounds for recall and
%precision given estimates of the number of validation errors. 
% (for example,
%using the 1/1000 rule-of-thumb to estimate the total number of SNPs would ignore important
%sample-specific information.it should depend for instance on whether the sample is
%European or African, since hg19 is European.
% the case of only positive labels and the case of comprehensive
%validation data.  

To quantify these errors, we first note that each called variant can be described
via two sets of labels: positive or
negative depending solely on the caller, and true or false depending also on
the validation dataset. We can compute true positives (TP), false positives (FP) and false
negatives (FN) from these two sets of labels. Proposition~\ref{prop:real_error} presents bounds on recall
and precision given $E$ and in terms of TP, FP and FN.  These bounds hold generally for SNPs, indels and SVs
evaluation, and for various evaluation metrics, e.g., metrics considering
zygosity, insertion sequence, etc. (see
%Section~\ref{sec:noisy_validation_bounds} in the Supplementary Material for further details and proof).
Section~D in the Supplementary Material for further details and proof).

\begin{prop}
\label{prop:real_error}
Let $\present = \TP + \FN$ be the number of positive labels, and let $P = \TP + \FP$
be the number of positive calls.  In the case of only positive validated
labels,
  \[
  \frac{\TP - E}{\present} \leq \recall \leq
  \begin{cases}
    \frac{\TP + E}{\present} &\mbox{if } E \leq \FN \\
    \frac{\TP}{\present - E} &\mbox{otherwise.}
  \end{cases}
  \]
In the case of both positive and negative labels, the same recall bounds apply,
and the following precision bounds hold: 
\begin{equation}
\frac{\TP - E}{P} \leq \precision \leq \frac{\TP + E}{P} \,.
\end{equation}
\end{prop}

Proposition~\ref{prop:real_error} states that recall and precision
have worst-case additive errors of the form $E/\present$ and $E/P$,
respectively.  We use these bounds when reporting results in Section~\ref{section:evaluation}. 
%and  of the fact that
%recall is being estimated from only a subset of the actual variants.  Thus, we
%are implicitly assuming that we can accurately estimate recall from the sample
%of variants present in the validation dataset.

\subsubsection{Ambiguity Resolution}
\label{ssec:ambiguity}
\ourbm incorporates three steps in the evaluation process  in order to minimize
the impact of VCF ambiguity. The first two steps, \emph{cleaning} and
\emph{left normalization}, are (standard) VCF preprocessing steps that we
perform independently on both the ground truth VCF and the predicted VCF.
%Cleaning involves making the VCF case-insensitive, removing homozygous
%reference calls and removing calls where the reference and alternate allele
%match.  
The cleaning step removes ambiguity associated with case discrepancies, and
also filters out extraneous VCF entries, i.e., homozygous reference calls and
calls where the reference and alternate allele match.  Left normalization
involves left shifting all variants as far as possible and is the VCF standard,
though this convention is not followed by all variant calling algorithms.  Left
normalization removes certain types of ambiguity associated with indels and SVs,
such as by unambiguously representing the deletion `GCGCGC $\rightarrow$ GCGC' as
a deletion event associated with the two leftmost `GC' bases. 

Our final step is a novel ambiguity resolution algorithm, \rescue, which 
involves a second pass over the VCF files during evaluation.  \markup{After initially
strictly comparing the calls between the predicted and the true VCFs, we aim to
`rescue' variants marked as incorrect (both false positives and false
negatives) due to VCF ambiguity. For each such call, we create two short
sequences by expanding the full sequence in some short window around the call
in the true and predicted VCF, respectively.  We then rescue the call if the
two sequences are equivalent. Rescued calls are thus by definition correct,
and notably, 
%Crucially, the calls rescued by \rescue are
%always correct, since these calls correspond to regions of the genome
%where the ground truth and predicted genomes are equivalent. 
\rescue can only improve the quality of the reported precision and recall
figures. Nonetheless, the number of calls that \rescue is able to rescue is
dependent on the window size, as discussed in
Section~\ref{ssec:amb_res_results}.}

%\markup{See Section~\ref{sec:ambiguity_details} in the Supplementary Material for
\markup{See Section~E in the Supplementary Material for
further
algorithmic details, including discussions about edge cases such as overlapping
alleles and combinations of alleles that cancel out.} 

\begin{table*}[ht]
  \centering
\caption{Effect of the window size parameter in the \rescue algorithm on indel
precision and recall for mpileup and GATK on the Venter genome. Error bounds are
  excluded since there is no uncertainty in the Venter validation data. }
\label{tab:winsize}
\small{
\begin{tabular}{c|cc|cc||cc|cc}
& \multicolumn{4}{c||}{mpileup} & \multicolumn{4}{c}{GATK} \\ 
& \multicolumn{2}{c}{Insertions} & \multicolumn{2}{c||}{Deletions} & \multicolumn{2}{c}{Insertions} & \multicolumn{2}{c}{Deletions} \\
\hline
Window & Prec & Rec & Prec & Rec & Prec & Rec & Prec & Rec \\ 
\hline
25 & 85.8\% & 70.9\%& 91.2\%& 75.2\%& 90.3\%& 87.8\%& 91.9\%& 90.6\%\\ 
50 & 85.8\%& 70.9\%& 91.3\%& 75.2\%& 90.6\%& 88.1\%& 92.3\%& 90.9\%\\ 
100 & 85.8\%& 70.9\%& 91.3\%& 75.2\%& 90.6\%& 88.1\%& 92.3\%& 90.9\%\\
150 & 85.7\%& 70.9\%& 91.3\%& 75.2\%& 90.5\%& 88.1\%& 92.2\%& 90.8\%
\end{tabular}
}
\end{table*}

\subsection{Usage}
All the materials necessary to run \ourbm are available at our website,
\ourwebsite. All relevant files are available for download, including
BWA-aligned BAM files containing the raw reads, ground truth VCF files and
reference files.  All scripts used to calculate results are available in a
public repository at \ourrepo, including evaluation, rescue, VCF normalization
and contamination scripts. \markup{All data included in \ourbm are derived from
publicly available sources and thus can be freely redistributed.}

We provide detailed instructions for running \ourbm on AWS. For a researcher
wishing to benchmark a new aligner, we describe how to download the
short reads, run the aligner on AWS, execute some or all of the variant callers
evaluated in Section~\ref{section:evaluation} and run the \ourbm evaluation
scripts to get performance and accuracy metrics.  In contrast, for a researcher
with a new variant caller, we describe how to download our aligned BAM files,
run the caller on AWS and evaluate the overall performance and accuracy.
%taking into account the precalculated cost and time
%of alignment to compute the full cost of this new pipeline.  Finally,  
%As noted above, it is important that benchmarks evolve to keep pace with the
%state of the art. 

Finally, we plan to update \ourbm as new validated datasets become
available. We also invite users to submit performance and accuracy results
associated with new aligner and variant caller pipelines to our results page.

\begin{table*}[ht!]
  \caption{Effect of ambiguity resolution on benchmarking GATK and mpileup on
  indels using the mouse dataset. The results illustrate the impact of each successive
  step of resolution, namely, cleaning, left normalization and rescuing.}
  \centering
  \label{tab:ambiguity_mouse}
\small{
\begin{tabular}[r]{c|cc|cc||cc|cc}
   & \multicolumn{4}{c||}{mpileup} & \multicolumn{4}{c}{GATK} \\
   & \multicolumn{2}{c}{Insertions} & \multicolumn{2}{c||}{Deletions} & \multicolumn{2}{c}{Insertions} & \multicolumn{2}{c}{Deletions} \\
  \hline
   Strategy & Pre & Rec & Pre & Rec & Pre & Rec & Pre & Rec \\
  \hline
  Cleaning & 81.1 $\pm$ 3.5 & 12.2 $\pm$ 0.5 & 75.2 $\pm$ 3.3 & 12.6 $\pm$ 0.6 & 73.1 $\pm$ 0.4 & 86.4 $\pm$ 0.5 & 68.9 $\pm$ 0.4 & 91.2 $\pm$ 0.6 \\ 
 Normalization & 76.6 $\pm$ 0.5 & 66.5 $\pm$ 0.4 & 76.6 $\pm$ 0.5 & 74.9 $\pm$ 0.5 & 85.7 $\pm$ 0.4 & 84.0 $\pm$ 0.4 & 80.5 $\pm$ 0.4 & 89.9 $\pm$ 0.5 \\ 
 \rescue & 87.8 $\pm$ 0.5 & 76.6 $\pm$ 0.4 & 79.0 $\pm$ 0.4 & 85.9 $\pm$ 0.5 & 92.0 $\pm$ 0.5 & 86.2 $\pm$ 0.4 & 85.5 $\pm$ 0.4 & 91.8 $\pm$ 0.5 \\
%    \rescuega (window = $alt1$) & $-$ & $-$ & $-$ & $-$ & $74.3 \pm 0.4$ & $94.3 \pm 0.5$ & $66.5 \pm 0.3$ & $96.3 \pm 0.5$ \\
%    \rescuega (window = $alt2$) & $-$ & $-$ & $-$ & $-$ & $74.3 \pm 0.4$ & $94.3 \pm 0.5$ & $66.5 \pm 0.3$ & $96.3 \pm 0.5$ \\
  \end{tabular}
}
\end{table*}

\begin{table*}[ht]
  \centering
\small{
  \caption{\markup{Benchmarking results for SNPs.}}
  \label{tab:SNP}
  \begin{tabular}[r]{r|rrrr||rrrr}
   dataset & \multicolumn{4}{c||}{mpileup} & \multicolumn{4}{c}{GATK} \\
  \hline
     & \multicolumn{1}{c}{Hours} & \multicolumn{1}{c}{Cost} & \multicolumn{1}{c}{Pre} & \multicolumn{1}{c||}{Rec} & \multicolumn{1}{c}{Hours} & \multicolumn{1}{c}{Cost} & \multicolumn{1}{c}{Pre} & \multicolumn{1}{c}{Rec}\\
  \hline
Venter & $2 h$ & $\$ 5$ & $98.7\% \pm 0.0$ & $97.0\% \pm 0.0$ & $57 h$ & $\$ 142$ & $99.3\% \pm 0.0$ & $91.7\% \pm 0.0$ \\
contam. Venter & $3 h$ & $\$ 8$ & $91.3\% \pm 0.0$ & $96.7\% \pm 0.0$ & $75 h$ & $\$188$ & $98.7\% \pm 0.0$ & $91.8\% \pm 0.0$ \\
NA12878 & $5 h$ & $\$13$ &  -  & $98.8\% \pm 0.0$ & $86 h$ & $\$215$ &  -  & $98.8\% \pm 0.0$ \\
contam. NA12878 & $5 h$ & $\$13$ &  -  & $98.8\% \pm 0.0$ & $110 h$ & $\$275$ &  -  & $98.8\% \pm 0.0$ \\
NA18507 & $4 h$ & $\$10$ &  -  & $99.0\% \pm 0.0$ & $154 h$ & $\$385$ &  -  & $99.0\% \pm 0.0$ \\
NA19240 & $4 h$ & $\$10$ &  -  & $98.7\% \pm 0.0$ & $167 h$ & $\$418$ &  -  & $99.0\% \pm 0.0$ \\
mouse & $6 h$ & $\$15$ & $98.4\% \pm 0.2$ & $87.3\% \pm 0.2$ & $107 h$ & $\$268$ & $97.8\% \pm 0.2$ & $94.9\% \pm 0.2$ \\
contam. mouse & $5 h$ & $\$13$ & $98.3\% \pm 0.2$ & $86.7\% \pm 0.2$ & $96 h$ & $\$240$ & $97.9\% \pm 0.2$ & $94.6\% \pm 0.2$ \\ 
\end{tabular}
}
\end{table*}
\begin{table*}
  \centering
\small{
  \caption{\markup{Benchmarking results for small deletions (excluding Pindel results).}}
  \label{tab:indel_del}
  \begin{tabular}[r]{r|rrrr||rrrr}
   dataset & \multicolumn{4}{c||}{mpileup} & \multicolumn{4}{c}{GATK} \\
  \hline
     & \multicolumn{1}{c}{Hours} & \multicolumn{1}{c}{Cost} & \multicolumn{1}{c}{Pre} & \multicolumn{1}{c||}{Rec} &\multicolumn{1}{c}{Hours} & \multicolumn{1}{c}{Cost} & \multicolumn{1}{c}{Pre} & \multicolumn{1}{c}{Rec} \\
  \hline
Venter & $2 h$ & $\$ 5$ & $91.3\% \pm 0.0$ &  $75.2\% \pm 0.0$ & $57 h$ & $\$ 142$ & $92.3\% \pm 0.0$ & $90.9\% \pm 0.0$ \\
contam. Venter & $3 h$ & $\$ 8$ & $91.7\% \pm 0.0$ &  $71.7\% \pm 0.0$ & $75 h$ & $\$188$ & $92.4\% \pm 0.0$ & $90.5\% \pm 0.0$ \\
mouse & $6 h$ & $\$15$ & $79.0\% \pm 0.4$ &  $85.9\% \pm 0.4$ & $107 h$ & $\$268$ & $81.5\% \pm 0.4$ & $95.8\% \pm 0.4$ \\
contam. mouse & $5 h$ & $\$13$ & $80.4\% \pm 0.4$ &  $84.9\% \pm 0.4$ & $96 h$ & $\$240$ & $82.8\% \pm 0.4$ & $95.6\% \pm 0.4$ \\
\end{tabular}
}
\end{table*}

\begin{table*}[ht!]
\small{
  \caption{\markup{Benchmarking results for long deletion (approximate evaluation with breakpoint tolerance of 100bp). }}
  \centering
  \label{tab:sv_del}
  \begin{tabular}[r]{r|rrrr||rrrr}
   dataset & \multicolumn{4}{c||}{Pindel} & \multicolumn{4}{c}{BreakDancer} \\
  \hline
     & \multicolumn{1}{c}{Hours} & \multicolumn{1}{c}{Cost} & \multicolumn{1}{c}{Pre} & \multicolumn{1}{c||}{Rec} & \multicolumn{1}{c}{Hours} & \multicolumn{1}{c}{Cost} & \multicolumn{1}{c}{Pre} & \multicolumn{1}{c}{Rec}\\
  \hline
Venter & $80 h$ & $\$ 200$ & $14.2\% \pm 0.0$ & $43.3\% \pm 0.0$& $2 h$ & $\$ 5$ & $6.9\% \pm 0.0$ & $4.9\% \pm 0.0$ \\
contam. Venter & \tooslow & \toomuch & - & - & $1 h$ & $\$ 2$ & $0.0\% \pm 0.0$ & $0.0\% \pm 0.0$ \\
NA12878 & $81 h$ & $\$203$ &  -  & $78.3\% \pm 2.9$& $3 h$ & $\$7$ &  -  & $95.7\% \pm 2.9$ \\
contam. NA12878 & \tooslow & \toomuch & - & - & \tooslow & \toomuch & - & -  \\
NA18507 & $168 h$ & $\$420$ &  -  & $60.0\% \pm 3.0$& $3 h$ & $\$7$ &  -  & $60.0\% \pm 3.0$ \\
NA19240 & $221 h$ & $\$ 552$ &  -  & $70.7\% \pm 2.8$& $29 h$ & $\$73$ &  -  & $77.6\% \pm 2.8$ \\
mouse & \tooslow & \toomuch & - & - & $4 h$ & $\$10$ & $18.1\% \pm 1.3$ & $11.6\% \pm 0.8$ \\
contam. mouse & $345 h$ & $\$ 862$ & $6.6\% \pm 4.8$ & $53.4\% \pm 0.8$& $3 h$ & $\$7$ & $18.7\% \pm 16.0$ & $1.0\% \pm 0.8$ \\ 
\end{tabular}
}
\end{table*}

\begin{table*}[ht!]
\small{
  \caption{\markup{Benchmarking results for long deletions (exact evaluation).}}
  \centering
  \label{tab:sv_del_0bp}
  \begin{tabular}[r]{r|rrrr||rrrr}
   dataset & \multicolumn{4}{c||}{Pindel} & \multicolumn{4}{c}{Breakdancer} \\
  \hline
     & \multicolumn{1}{c}{Hours} & \multicolumn{1}{c}{Cost} & \multicolumn{1}{c}{Pre} & \multicolumn{1}{c||}{Rec} & \multicolumn{1}{c}{Hours} & \multicolumn{1}{c}{Cost} & \multicolumn{1}{c}{Pre} & \multicolumn{1}{c}{Rec}\\
  \hline
Venter & $80 h$ & $\$200$ & $0.0\% \pm 0.0$ & $40.6\% \pm 0.0$& $2 h$ & $\$ 5$ & $0.0\% \pm 0.0$ & $0.2\% \pm 0.0$ \\
contaminated Venter & \tooslow & \toomuch & - & - & $1 h$ & $\$ 2$ & $0.0\% \pm 0.0$ & $0.0\% \pm 0.0$ \\
NA12878 & $81 h$ & $\$203$ &  -  & $78.3\% \pm 2.9$& $3 h$ & $\$7$ &  -  & $0.0\% \pm 2.9$ \\
contaminated NA12878 & \tooslow & \toomuch & - & - & \tooslow & \toomuch & - & -  \\
NA18507 & $168 h$ & $\$420$ &  -  & $60.0\% \pm 3.0$& $3 h$ & $\$7$ &  -  & $0.0\% \pm 3.0$ \\
NA19240 & $221 h$ & $\$ 552$ &  -  & $70.7\% \pm 2.8$& $29 h$ & $\$73$ &  -  & $5.2\% \pm 2.8$ \\
mouse & \tooslow & \toomuch & - & - & $4 h$ & $\$ 10$ & $0.0\% \pm 1.3$ & $0.1\% \pm 0.8$ \\
contaminated mouse & $345 h$ & $\$ 862$ & $0.0\% \pm 4.8$ & $52.3\% \pm 0.8$& $3 h$ & $\$7$ & $0.0\% \pm 16.0$ & $0.0\% \pm 0.8$ \\
\end{tabular}
}
\end{table*}

\begin{table*}[ht!]
\caption{\markup{Computational performance statistics for the Venter dataset.}}
  \centering
  \label{tab:venter_stats}
  \begin{tabular}[r]{r|rrrrrrr}
  caller & Clock Time & Cost & CPU Time & Max Threads & Max disk & Max memory & Avg. memory \\
  \hline
mpileup & $2 h$ & $\$5$ & $28 h$ & $24$ & $0.9$ GB & $59.1$ GB & $58.8$ GB \\
GATK & $57 h$ & $\$142$ & $179 h$ & $16$ & $530$ GB & $59.1$ GB & $53.4$ GB \\
Pindel & $80 h$ & $\$200$ & $80 h$ & $1$ & $3.4$ GB & $59.1$ GB & $58.7$ GB \\
BreakDancer & $2 h$ & $\$5$ & $2 h $ & $1$ & $0.02$ GB & $59.1$ GB & $58.7$ GB \\
  \end{tabular}
\end{table*}

\section{Results}
\label{section:evaluation}

In this section, we first evaluate the impact of our ambiguity resolution
algorithms. We next illustrate the utility of \ourbm by evaluating the
performance of some leading SNP, indel, and structural variant calling
algorithms.  The goal of these experiments is to highlight \ourbm's
functionality, and to simplify the discussion, we use the default settings for
all variant calling algorithms and the same machine instance in all of our EC2
experiments. We note that improved accuracy and computational performance results
may indeed be possible via parameter tuning and optimizing to minimize EC2
instance footprints.

In all reported results, evaluation is reported as described in
%Section~\ref{ssec:metrics}. See Section~\ref{sec:experimental_details} in the
Section~\ref{ssec:metrics}. See Section~G in the
Supplementary Material for further implementation details.

\subsection{Ambiguity Resolution}
\label{ssec:amb_res_results}

%Window size is a user-specified parameter in the \rescue algorithm and, 
We first examine the precision and recall rates for mpileup and GATK as a
function of this user-specified window parameter, as Table~\ref{tab:winsize}
illustrates.
%with a range of window settings.  
We see that precision and recall are quite robust to various window sizes.
Nonetheless, for small windows ($\le25$bp), comparatively fewer true positives
are rescued, likely due to the omission of variants associated with an
ambiguously represented event. Meanwhile, large windows ($\ge150$bp) also lead
to fewer rescues, likely due to nearby false positives unrelated to the
ambiguously represented event entering the window.  We observe that a $50-100$bp
window balances these two tradeoffs, and we use a $50$bp window in all
subsequent experiments for computational reasons.

Next, Table~\ref{tab:ambiguity_mouse} shows the impact of each of the three
ambiguity resolution steps. The results show that all three steps significantly
impact the precision and recall of both GATK~\citep{DePristo11} and
mpileup~\citep{Li09} on calling indels for the mouse dataset. Our ambiguity
resolution has a similar impact on indel detection for the Venter dataset, and
also has a significant (though less drastic) impact on SNP detection for both
%datasets (see Section~\ref{sec:ambiguity_details} in the Supplementary Material
datasets (see Section~E in the Supplementary Material
for further results).

\subsection{SNP Calling}

We benchmark the performance of GATK and mpileup to call SNPs, 
%We evaluate SNP calls by exact location and unphased genotype.
and Table~\ref{tab:SNP} summarizes the results.  The results show that
%\item GATK's high recall on NA12878 may be partially explained by GATK's use of
%HapMap as a prior, since our NA12878 SNP dataset is a subset of HapMap.
%\item For GATK on mouse, we were not able to run variant recalibration due to
%lack of overlap between dbSNPs variants and our reference.  For the human
%datasets, GATK's prior is more informed, including HapMap and other sources of
%known variation.
GATK is more computationally expensive, but does not strictly outperform
mpileup on the uncontaminated datasets.
%\footnote{With GATK it is
%possible to trade off between recall and precision by altering the stringency
%of variant filtering, though we did not attempt such tuning in this work.}
Moreover, the effect of contamination on precision is fairly visible on Venter,
where mpileup's accuracy clearly degrades while GATK appears robust to the
contamination.  The difference in performance on the mouse and contaminated
%mouse dataset is less pronounced, since, as noted in Section~\ref{sec:align}
mouse dataset is less pronounced, since, as noted in Section~F
(Supplementary Material), the aligner was able to filter the contaminated reads
before they were processed by mpileup or GATK.

\subsection{Indel Calling}
\label{ssec:indel_eval}

%We define ``indels'' as 
%We restrict attention to pure insertions and deletions.  The most lenient way
%of judging an indel call is by considering only two integers: the (left)
%breakpoint, and the length of the inserted or deleted sequence.  For both
%integers, we allow some ``wiggle room'', which for simplicity we choose to be
%the same.
We evaluate the performance of mpileup, GATK and Pindel on the detection of
indels, in particular insertions or deletions 50bp or less, with detailed
results presented in \markup{Table~\ref{tab:indel_del}},
%Table~\ref{tab:indel_del_all}, and Table~\ref{tab:indel_ins}.\footnote{\ourbm's
Table~S6, and Table~S7.\footnote{\ourbm's
sampled human benchmarking datasets do not include any validated indels, and so
the indel results are restricted to the mouse and synthetic human datasets.}
These results demonstrate that GATK and Pindel outperform mpileup in terms of
accuracy, but are more expensive computationally.  In fact, Pindel highlights
the importance of \ourbm's computational performance metrics, as it failed to
complete within our predetermined time limit of 400 hours (i.e., a $\$1000$ AWS
budget) on both the contaminated Venter and mouse datasets, and we thus did not
obtain results for these experiments. On both contaminated Venter and
contaminated mouse, mpileup and GATK gained slightly in precision and worsened
in recall as compared to their uncontaminated counterparts. This difference
seems to result from the fact that both algorithms predicted fewer indels
overall on the contaminated sets than they did SNPs; for example, mpileup
called 2.5\% fewer indel deletions on mouse and only 0.5\% fewer SNPs.
Consistent with the SNP results, however, mpileup was less robust to
contamination than GATK.
%\item We do not report BreakDancer's performance on indels, as it is designed
%  for variants larger than 100bp.  By BreakDancer we mean BreakDancerMax; there
%  is another tool, BreakDancerMini, better suited to indels, but we did not try
%  it.

\subsection{Structural Variant Calling}
We benchmark Pindel and BreakDancerMax on insertions and deletions greater than
50bp in length, reporting \markup{approximate breakpoint  
%results} in Table~\ref{tab:sv_del} and Table~\ref{tab:sv_ins} (Supplementary Material).
results} in Table~\ref{tab:sv_del} and Table~S8 (Supplementary Material).
We first note that we do not report results for the several experiments that
ran longer than our 400 hour budget, namely BreakDancer on contaminated NA12878
and Pindel on mouse, contaminated Venter and contaminated NA12878.  On the
experiments that did complete, we observe that (perhaps unsurprisingly) SV
accuracy is much lower than SNP and indel accuracies, and in particular both
Pindel and BreakDancer have fairly low accuracy on long insertions.  Indeed,
for NA12827, BreakDancer misses all insertions even though it accurately
identifies more than half of the deletions.  Moreover, we observe that
BreakDancer's recall is much higher for the sampled human datasets than for
Venter and mouse. This discrepancy can be explained by the fact that, unlike
the sampled NA12878 validation data, the comprehensive Venter and mouse
datasets contain many short structural deletions, and BreakDancer is not
designed for shorter variants. 

\markup{We observe drastically different results when evaluating the exact
breakpoint of accuracy of these methods (see Table~\ref{tab:sv_del_0bp} and
%Table~\ref{tab:sv_ins_0bp}).  Indeed, both algorithms suffer significant
Table~S9).  Indeed, both algorithms suffer significant
decreases in accuracy under this more stringent evaluation metric. The
precision for both algorithms on all applicable datasets approaches
zero, and BreakDancer makes almost no exactly accurate calls for insertions or
deletions. Pindel, however, does identify some long insertions and deletions by
exact breakpoint, thus yielding a minimal degree of recall.}

% \todo{(add second table with results stratified by SV length, perhaps in appendix)}

%Table~\ref{tab:SV}
%illustrates this point by presenting the results of both algorithms on `long'
%structural deletions (defined as 1kb or longer), showing drastically improved
%performance by BreakDancer on Venter and mouse.
%\item BreakDancer performs quite poorly on the long-insert reads.  Perhaps
%neither BWA nor BreakDancer were designed for reads with insert size in the
%kilobases, such as the mouse long-insert reads.
%Performance for insertions on both Pindel and Breakdancer is fairly low. Pindel is not designed
%for them, while BreakDancer finds very few structural insertions 
%in comparison to structural deletions.
%, perhaps  because the coverage
%is at most 30x in these datasets and insertions are inherently harder to find
%than deletions. 
%Surprisingly, BreakDancer 
%Indeed, for NA12827, Breakdancer misses all insertions even though it
%accurately identifies more than half of the deletions.
 %There are the same number of deletions and insertions in this dataset, 22.
%\item TODO: Why does BreakDancer run so fast on Venter?

Finally, we note that neither caller handles the contaminated datasets
particularly well, though with somewhat different modes of failure. Pindel
suffered computationally when dealing with the contaminated datasets, failing to
complete on the contaminated Venter and the contaminated NA12878 datasets, and
required 375 hours to process the contaminated mouse dataset.  Although
BreakDancer did not complete on the contaminated NA12878 dataset, it in fact
executed very quickly on both the contaminated Venter and contaminated mouse datasets.
However, its accuracy suffered greatly relative to its accuracy on the analogous
non-contaminated datasets, as it found almost no variants in either case.

\subsection{\markup{Computational Performance}}
\markup{Although the cost of running a given caller on Amazon's AWS platform provides a
convenient single metric for comparison, we also provide more fine-grained
computational performance metrics to highlight differences between the callers
and to help researchers optimize their choice of computational platforms. In
Table~\ref{tab:venter_stats} we present the performance metrics for all four
%callers on the Venter dataset (see Table~\ref{tab:perf_stats} in Supplementary Material for
callers on the Venter dataset (see Table~S10 in Supplementary Material for
statistics on other datasets). All four callers use virtually all 60.5GB of memory available
to them on Amazon's cc2.8xlarge instance; mpileup, Pindel, and BreakDancer do
so consistently through their runs, but GATK's memory usage fluctuates as
shown by its lower average memory usage, most likely because only portions of
the GATK pipeline are multi-threaded.  GATK also requires a large amount of
disk space, while mpileup and Pindel use only modest amounts; since
BreakDancer's output is in a more compact format than VCF, it requires almost
none.  We also note that since Pindel became memory-bound on our chosen
instance type, we ran it single-threaded; we thus ran BreakDancer on a single
thread for consistency.}

\section{Discussion}

Hundreds of variant calling algorithms have been proposed, and the majority of
these algorithms have been benchmarked in some form (see detailed discussion in
%Section~\ref{sec:related} in the Supplementary Material). To the best of our
Section~A in the Supplementary Material). To the best of our
knowledge, none of these existing benchmarking methodologies account for noise
in validation data, ambiguity in variant representation, or computational
efficiency of variant calling methods in a consistent and principled fashion.
Given the rapid growth of next-generation sequencing data, the need for a
robust and standardized methodology has never been greater.

The IT industry serves as an illuminating case-study in the context of
benchmarking. Similar to next-generation sequencing technology, the IT industry
has benefited greatly from the additional hardware resources provided by
Moore's Law. However, the industry's rapid progress also hinged on
the agreement on proper metrics to measure performance as well
as consensus regarding the best benchmarks to run to fairly evaluate competing
systems \citep{Patterson12}. Prior to this industry-wide agreement, each
company invented its own metrics and ran its own set of benchmarking evaluations,
making the results incomparable and customers suspicious of them.  Even worse,
engineers at competing companies were unable to determine the usefulness of
their competitors' innovations, and so the competition to improve performance
occurred only within companies rather than between them. Once the IT industry
agreed on a fair playing field, progress accelerated as engineers could see
which ideas worked well (and which did not), and new techniques were developed
to build upon promising approaches.
% note: We should perhaps say something about how care must be taken to create
% robust benchmarking data and proper metrics, as new techniques optimize for
% the benchmark. The UCI datasets are a good example of the danger of
% overfitting to a benchmark, as described in a vision paper at ICML 2012. 

Similarly, we believe that \ourbm could help accelerate progress in the field of
genomic variant calling.  We have compiled a rich collection of datasets and
developed a principled set of evaluation metrics that together allows for
quantitative evaluation of variant calling algorithms in terms of accuracy,
computational efficiency and robustness/ease-of-use (via ability to run on
AWS). \markup{Moreover, although \ourbm currently focuses on benchmarking variant
calling algorithms for normal human genomes, we believe that the motivating
ideas behind \ourbm, along with the tools developed as part of \ourbm, will be
useful in devising analogous variant calling benchmarking toolkits for human
cancer genomes and for the genomes of other organisms.}

Finally, we view \ourbm as a work in progress, as the contents of
\ourbm reflect (and are limited by) existing technologies. \ourbm currently has
limited ground truth data for human genomes, and the validation data across
datasets are enriched in `easier' non-repetitive regions due to underlying
sequencing and chip biases. Like any benchmarking suite, \ourbm must evolve
over time in order to stay relevant.  As new sources of validation data become
available, e.g., the NA12878 knowledge base \citep{Knowledgebase} or curated variants from
Illumina platinum genome trios \citep{platinum}, these datasets should be incorporated
into \ourbm.  Existing datasets should also be updated to keep them fresh and
prevent algorithms from `overfitting' to stale benchmarks, and benchmarking
datasets should be deprecated as new data sources obviate their utility.

%In the
%setting of cancer genomics, the benchmarking data would involve tumor-normal
%pairs, and evaluation would focus on detection of (rare) somatic mutations and
%the effect of purity / heterogeneity of tumor samples on predictive accuracy.

\section*{Acknowledgement}
This research was supported in part by NSF awards 1122732 and DBI-0846015, NIH
National Research Service Award Trainee appointment on T32-HG00047, NSF CISE
Expeditions award CCF-1139158, DARPA XData Award FA8750-12-2-0331, and gifts
from Amazon Web Services, Google, SAP,  Cisco, Clearstory Data, Cloudera,
Ericsson, Facebook, FitWave, General Electric, Hortonworks, Huawei, Intel,
Microsoft, NetApp, Oracle, Samsung, Splunk, VMware, WANdisco and Yahoo!.  

We thank David Bentley, Bill Bolosky, Mauricio Carneiro, Mark Depristo, Michael
Eberle, Adam Ewing, Gaddy Getz, David Haussler, Jeff Kidd, Jon Kuroda, Elliott Margulies,
Jim Mullikin, Frank Nothaft, Ravi Pandya, Benedict Paten, Taylor Sittler, Arun
Wiita, Kai Ye, and Matei Zaharia for useful insights. 
    
%\paragraph{Funding} 

\bibliographystyle{plainnat}
\bibliography{benchmarking_refs}

%\end{document}
\clearpage
\appendix

\renewcommand{\thesection}{\Alph{section}}
\setcounter{figure}{0}
\setcounter{table}{0}
\makeatletter 
\renewcommand{\thetable}{S\@arabic\c@table}
\renewcommand{\thefigure}{S\@arabic\c@figure}
\makeatother

\section{Extended Survey of Related Work}
\label{sec:related}
Many variant calling algorithms have been proposed, and the majority of these algorithms
have been benchmarked in some form. To the best of our knowledge, none of these
existing benchmarking methodologies account for noise in validation data,
ambiguity in variant representation, or computational efficiency of variant
calling methods in a consistent and principled fashion. Although it is not
practical to list all of these past works, in this section we highlight some
notable studies that illustrate the current state of benchmarking. 

The 1000Genomes Project \citepappendix{1kGenomes10} calls
%The 1000Genomes Project \citepappendix{1kGenomes10, Consortium12} calls
SNPs, indels, and structural variants (SVs) using a rich variety of
computational approaches and reports the consensus variant calls of these
approaches. The project also invests significant resources to estimate false positive
rates by generating an ad-hoc collection of validation data using orthogonal
technologies. In contrast to \ourbm, this approach is aimed at
validating variants called by specific algorithms, and moreover, the validation
efforts of this project vary significantly across algorithm and problem type
(e.g., SNPs called from exome sequencing are validated
using a different methodology than SVs called from high-coverage full genome sequencing).  

The Archon Genomics X PRIZE validation protocol \citepappendix{Kedes11}
proposes a consensus approach to evaluate variant callers, utilizing a variety
of technologies including next-generation sequencing, microarrays, and Sanger
sequencing.  The developers of the CloudBreak SV caller
\citepappendix{Whelan13} benchmark their algorithm against a synthetic genome
constructed from the Venter genome \citepappendix{Levy07}, similar to one of
\ourbm's simulated datasets.  They also benchmark their calls against the union
of SVs called from three earlier studies on a real human, without quantifying
the inherent uncertainty in these calls or reconciling differences amongst
them. \citeappendix{David13} also proposes the simulation of a fake reference,
though in the context of parameter selection for their algorithm for Alu
detection.

The Genome in a Bottle Consortium (GBC) \citepappendix{Zook11} has similar
goals as \ourbm, though they have taken a different approach, as they evaluate
end-to-end pipelines starting from DNA, target accuracy metrics, and
perform evaluations against one specific genome. The Genome Comparison and
Analytic Testing resource (GCAT) \citepappendix{GCAT} in part leverages the
validation set from GBC, and performs evaluation of both alignment and variant
calling methods using consensus-based and crowdsourced evaluation.

\citeappendix{DePristo11} evaluate GATK SNP calling
by measuring concordance with HapMap3 \citepappendix{Consortium10} and 1000Genomes
\citepappendix{1kGenomes10}, and by measuring known/new and Ti/Tv ratios.  HapMap SNPs
\citepappendix{Consortium10, Frazer07, Consortium05} are subjected to extensive quality
%\citepappendix{Consortium10, Frazer07} are subjected to extensive quality
assurance, mainly by measuring consensus across different technologies (SNP
chips, PCR sequencing for the ENCODE project, and SNPs from fosmid end
sequencing \citepappendix{Kidd08}), along with a battery of quality control filters
largely specific to the technologies involved.

Additional studies have proposed benchmarking techniques for related, yet
distinct, problems.  The authors of FRCbam \citepappendix{Vezzi12} propose methods to
assess overall de novo assembly quality and correctness, thus addressing a
different problem than the variant calling benchmarking problem that we tackle.
The authors of MuTect \citepappendix{Cibulskis13}, a tool for detecting somatic SNPs in
tumor cells, eschew simulated data in favor of benchmarking using downsampling
(randomly discarding reads from previous validated datasets to obtain desired
coverage) and 'virtual tumors' (sequencing the same sample twice to test false positive rate, 
and artificially adding SNPs to real reads to test false negative rate).  These techniques are
interesting, though they are more applicable to the tumor-normal SNP problem
than that of detecting germline mutations (in particular, for germline indels
and SVs, downsampling is seldom necessary and artificially altering reads to
create SV signatures would be more challenging than doing so for SNPs).

Various studies have also publicly released experimentally validated variant
calls for a variety of genomes.  These studies provide potentially rich
sources of data that could be useful for benchmarking, but the error rates of
these variants cannot be easily quantified,
%most variants provided
%cannot be trusted to be more than $99\%$ reliable, 
and in the case of indels
and SVs, most breakpoints are not called to nucleotide resolution. Hence we
chose to not incorporate these datasets in the initial version of \ourbm.  Just
a few examples of such validation data can be found in the following
references: SNPs: HapMap3 \citepappendix{Consortium10, Mccarroll08}; indels:
%references: SNPs: HapMap3 \citepappendix{Consortium10}; indels:
\citeappendix{Mills11_insdel}; SVs: \citeappendix{Mccarroll08, Mills11_CNV, Mills11_insdel};
%\citeappendix{Mills11_insdel}; SVs: \citeappendix{Mills11_CNV, Mills11_insdel};
SNPs/indels/SVs: \citeappendix{Kidd08}, 1000Genomes \citepappendix{1kGenomes10}.
%We chose to not incorporate these datasets in the initial version of \ourbm
%since they do not fit our ideal of gold-standard benchmarking.
% yet they could indeed be added to our benchmarking toolkit in the future.

% \begin{itemize}
% \item hapmap, Mullikin: rich source of data, but not curated, cross referenced, or quantified for error
% \item other papers we didnt use b/c SV calls arent precise (though could be added in the future) (Jesse list them!) 
% \item any others??
% \end{itemize}
%As part of this work, identifying the most reliable subset of this data, and
%curate this subset of data by quantifying its reliability and making it
%interpretable and availble benchmarking purposes.

% \begin{itemize}
% \item 1K genomes: consensus
% \item XPrize protocal: consensus (see kristal's email)
% \item Platinum genome (not based on orthogonal validation)
% \item ohsu paper: they created a single experimental genome? 
% \item references ohsu paper mentioned?
% \item Lumpy SV paper: use synthetic data with Venter (Ravi's link)
% \item Genome in a Bottle (http://genomeinabottle.org/blog-entry/existing-and-future-na12878-datasets)
% \item http://www.nyu.edu/about/news-publications/news/2013/01/02/researchers-develop-tool-to-evaluate-genome-sequencing-method.html
% \end{itemize}

\section{Data Preparation}
\subsection{Synthetic Datasets}

We derive our synthetic datasets from J. Craig Venter's genome (HuRef).  We did
not use the HuRef sequence directly due to the difficulty of computing a most
parsimonious VCF ``diff'' between two full-genome assemblies; optimal alignment
algorithms do not scale well enough, and approximate alignment algorithms would
not be suitable for benchmarking.  Instead, we create the sample genome by
applying the HuRef variants provided by \citeappendix{Levy07} after `lifting
over' \citepappendix{Hinrichs06} from hg18 to hg19. This genome features 0.8
million indels and structural insertions and deletions, and 90 inversions.
Since all variants in this set come from \citeappendix{Levy07}, this set lacks
much of HuRef's structural variation \citepappendix{Pang10}.  

We generate reads from the sample genome using simNGS with its default
settings.  After invoking simNGS, each simulated read is annotated with its
associated location in the reference, as well as its location of origin in the
sample genome. 

To contaminate the Venter reads (Section~\ref{ssec:corruption_details}), we
introduced simulate reads from an approximation of James Watson's genome. We
used the variants of dbSNP population 12269,
% <http://www.ncbi.nlm.nih.gov/projects/SNP/snp_viewTable.cgi?pop=12269>
found in Watson, with additional random SNPs sprinkled in, with
heterozygous/homozygous ratio 1.2 \citepappendix{Levy07} and Ti/Tv ratio 2.1, yielding
a 90/10 known/new ratio.  We also introduced random structural insertions and
deletions of length 1000, yielding 20 megabases of structurally variant
sequence \citepappendix{Pang10}.

%To corrupt the Venter reads (Section~\ref{ssec:corruption_details}), we 
%\tvsim{} to simulate reads from an approximation of James Watson's genome.  The
%same \tvsim{} settings were used for Watson as for Venter, except regarding
%introduction of variants.  We used the variants of dbSNP population 12269,
% <http://www.ncbi.nlm.nih.gov/projects/SNP/snp_viewTable.cgi?pop=12269>
%found in Watson, with additional random SNPs sprinkled in, with
%heterozygous/homozygous ratio 1.2 \citepappendix{Levy07} and Ti/Tv ratio 2.1, yielding
%a 90/10 known/new ratio.  We also introduced random structural insertions and
%deletions of length 1000, yielding 20 megabases of structurally variant
%sequence \citepappendix{Pang10}.

% The contaminated synthetic dataset was created from the Venter and Watson short-reads
% as described in Section~\ref{ssec:corruption_details}.

\subsection{Mouse Datasets}
\label{app:mouse_prep}
The mouse datasets leverage existing mouse genomic data associated with the
canonical mouse reference as well as from the Mouse Genomes Project (MGP).  Our
first dataset is constructed as follows:
\begin{enumerate}
\item Use the canonical homozygous mouse reference as a sample, not as a
reference.  The canonical mouse reference \citepappendix{Church09} (known as `mm10') is
derived from the homozygous mouse strain C57BL/6J (abbreviated `B6').  In
\ourbm, we use mm10 as an extremely well-sequenced sample.  
\item Leverage existing short-read data for the B6 strain \citepappendix{Gnerre11}. 
% We associate several sets of short reads with this dataset, all of which come
% from \citeappendix{Gnerre11}.
Compared with the synthetic and human reads, the mouse reads have a relatively
high error rate, as discussed in more detail in Section~\ref{sec:align}, but
given their high coverage %and multiple insert sizes, 
they are still useful for evaluation of variant calling algorithms.
There are ten libraries of paired-end Illumina reads, all from a B6 female.  In
this paper we consider only the highest-coverage library: 101bp
overlapping ends with 58.6x coverage and mean fragment size of 168.
% and 101bp ends with 20.0x coverage and fragment size 1.8k.

% There are four sets of paired-end Illumina reads, all from a B6 female.
% The first set consists of 101bp overlapping ends with 58.6x coverage
% and mean fragment size of 168.  The remaining three sets are long-insert:
% 101bp ends with 48.0x coverage and fragment size 2.2k;
% 26bp ends with 13.5x coverage and fragment size 7.5k;
% and 76bp ends with 1.4x coverage and fragment size 38k.
\item Create a reference by applying variation from the Mouse Genomes Project
(SNPs, indels, SVs). We leverage a rough set of variants for the DBA mouse
strain that are called relative to the mm10 reference. In particular, we use
variation from the Mouse Genomes Project consisting of 28K structural variants
(SVs), consisting of insertions, deletions, and inversions; 0.9 million indels;
and 5.6 million SNPs. After `lifting over' these variants from mm9 to mm10, we
create a `fake reference' in which DBA variants are injected into the mm10
reference, and also generate a VCF file storing mm10 variants relative to this
`fake reference.' 
%is created that stores the set of DBA
%variants relative to mm10.  
%the fake reference  we created is a homozygous ``fake reference'', representing
%a female of the DBA strain, 
%\item Generate the final VCF file storing mm10 variants.
%Invert the VCF file generated by \tvsim. that variation to yield validation
%data.
%We invert the VCF file created by \tvsim to obtain a set of variants of mm10
%relative to DBA. 
%We then inverted this variation (after lifting over from mm9 to mm10),
%yielding validation data in VCF format, for B6 relative to the fake reference.
\end{enumerate}

Moreover, our second mouse dataset is a contaminated version of this initial
dataset, created from the mouse reads as well as corresponding human reads.
The human reads, for NA12878, were generated by the same sequencing methodology
\citepappendix{Gnerre11}, with fragment size mean and standard deviation 155 and 26,
versus 168 and 32 for the mouse.  See Section~\ref{ssec:corruption_details} in
the Supplementary Material for further details.

\paragraph{\textbf{Detailed Validation Error Profile:}}
% Ameet's original note here:
% FILL IN; Error is due to discrepancy between reference and the short reads; two
% sources for this discrepancy: i) mutations introduced in subsequent generations
% of these homozygous mice (as discussed in the HYDRA paper); ii) issues with the
% BROAD reads (hopefully we can quantify this using alignment stats for the mouse
% and comparing them to stats for human reads).  
The first potential source of error is the mm10 assembly itself.  On the whole
it is ``of high fidelity and completeness, and its quality is comparable to, or
perhaps better than, that of the reference human genome assembly''
\citepappendix{Church09}.
\citeappendix{Church09} and \citeappendix{Quinlan10} discuss how the mouse reference is
imperfect in the $5\%$ of the genome that consists of segmental duplication.
However, our validation data (DBA variation extracted from MGP) contain no
variants in these highly repetitive regions, and thus only precision
(see Section~\ref{ssec:metrics} for discussion of evaluation metrics) could
%(see Section~2.2 for discussion of evaluation metrics) could
potentially be affected by these duplications.
% We sidestep this issue by screening off regions of segmental duplication when
% estimating precision.

%However, these repetitive regions are a subset of the ''inaccessible regions''
%introduced by \citeappendix{Keane11}, which correspond to roughly 15\% of the genome
%that is not well-supported by short-reads due to high levels of mapping
%ambiguity.  Moreover,  our validation data (DBA variation extracted from MGP) 
%contains no variants in these  ``inaccessible regions'' that total 15\% of the
%genome, and likely include most of the regions where mm10 may be unreliable
%\citepappendix{Keane11}, .  As a results, variant calling algorithms that are able to
%call variants in these inaccessible regions may have a misleadingly high false
%positive rate on the mouse dataset.

Another possible source of error is that the individual from which reads were
taken may be slightly genetically different from the individual from which mm10
was constructed, despite both belonging to the inbred strain B6
\citepappendix{Watkins08}.  We provide back-of-the-envelope upper bounds for the number
of SNP, indel, and SV differences, by bounding both the number of intervening
generations and the mutation rate per generation.  To bound the former, note
that both mice came from the Jaxson laboratory, whose B6 stock at time of
writing \citepappendix{Jackson} is 235 generations from strain origin.  Since the
strain is maintained by sibling mating, and the two mice are far more recent
than strain origin, then conservatively speaking, there have been
fewer than 470 opportunities for generational mutation to strike.  The SNP
mutation rate is about $10^{-8}$ per generation \citepappendix{Egan07}.  To
estimate mutation rates for indels and SVs, we use the approximation of
\citeappendix{Cartwright09} that for every 100 SNPs, there are at most 16 indels or
structural insertions or deletions, whose lengths follow a zeta distribution
with parameter roughly $1.7$.  Multiplying the resulting mutation rates, the
generational gap, and the size of mm10 yields
upper bounds of 13k SNPs, 2k indels, and 93 structural insertions and
deletions (corresponding to $0.2\%$ error rates for SNPs and indels, and a
$0.3\%$ error rate for SVs), affecting the reads in our mouse dataset but not
present in the validation data.

\subsection{Sampled Human Datasets}
\label{app:human_prep}

\paragraph{\textbf{Detailed Validation Error Profile:}}

For our validated SNPs, we have indirect evidence of the error rates of
each of the two chip technologies.  Perlegen has a 0.51\% rate of discrepancy
against other HapMap2 calls (Supplementary text 2 of \citeappendix{Frazer07}).
%Supplementary Table 3 of \citeappendix{Frazer07} provides discrepancy rates of the
%HapMap2 chip technologies against fosmid end SNPs: the rates are $1.43\%$ for
%Perlegen, while for BeadArray they vary according to sequencing center, from
%$0.06\%$ to $1.59\%$.  Fosmid end SNPs are obtained according to a procedure
%orthogonal to SNP chips, so it is reasonably safe to treat their errors as
%independent.  
Furthermore, \citeappendix{Oliphant02} claim that for SNP chips such as BeadArray,
``genotyping should have accuracy above $99\%$''.  Therefore, even with a
conservative assumption of $2\%$ error rates, and assuming that the errors from
the two chips are independent (reasonable since the two chips employ different
strategies), our validated SNP data would have error rate $0.04\%$.

%The fosmid SNPs used in our smaller SNP dataset have an error rate of $4\%$,
%which \citeappendix{Kidd08} estimate by comparing against PCR reads from the ENCODE
%project and manually inspecting 100 discrepancies.  This error rate combines
%multiplicatively with the error rate of our larger SNP dataset due to
%orthogonality of the technologies.

Moreover, since the SVs we parsed were called from finished fosmid sequence,
understanding the potential sources of error requires understanding the
intricate, partially manual process that generated these calls \citepappendix{Kidd08,
Kidd10insights}.  A fosmid is a $\sim$40kb fragment of DNA.  The Human Genome
Structural Variation Project \citepappendix{Waterston07} provides libraries of
end-sequenced fosmids for 22 individuals.  The fosmid ends, each corresponding
to a Sanger read, were mapped to the human reference genome.  Based on this
mapping, each fosmid was classified as \emph{concordant} or \emph{discordant},
the latter meaning potentially harboring a SV, for instance due to the ends
being close together, far apart, or in opposing orientation.  Imprecise SVs
were called from clusters of discordant fosmids. 
% 13 of the 22 people enjoy over 150 such fosmids.  For nine of
%these 13 people, some of these fully sequenced fosmids were analyzed for SV
%breakpoints. 
Selected discordant fosmids contributing to imprecise SV calls were fully
Sanger sequenced, and some of these fully sequenced fosmids were analyzed for
SV breakpoints. 
There are four possible sources of error in the final, precise SV calls:
\begin{enumerate}
  \item Calling imprecise SVs. Three measures were taken to ensure reliability
of these imprecise calls \citepappendix{Kidd08}.  First, discordant fosmids had to pass
a series of stringent mapping criteria.
%from \citeappendix{Kidd08}: ``In order to contribute to a discordant site, discordant
%clones had to pass more stringent mapping criteria: at least 150bp of
%non-repeatmasked bases included in the alignment ($2\%$ divergence threshold),
%a “best” map location (based on the previously described 13-point scoring
%system), and a read length of at least 400 bp. We further required that ESPs
%map with a minimum sequence identity of $99.5\%$ and include more than 30
%$basepairs of at least PHRED quality $>\mathrm{Q}30$.''  
Second, fosmids in a cluster had to show the same type of discordancy.  Third,
putative SV sites were orthogonally validated, e.g., with complete restriction
enzyme digests, microarrays, or trio testing.  \footnote{An even more reliable
dataset could be constructed by restricting attention to the approximately 400
SV sites that overlap with sites found by an orthogonal technology, arrayCGH
\citepappendix{Kidd08}.} 
\item End mapping of a discordant fosmid that is fully sequenced. Given the
existence of multiple discordant fosmid ends supporting a particular imprecise
SV, an incorrect mapping would require multiple SVs having nearly identical
flanking sequences.  However, this is unlikely given that the vast majority of
SVs have at least 20kb flanking sequences (and the shortest flanking sequence
was 8kb).  Additionally, an incorrect mapping would likely cause the breakpoint
calling step to fail, which \citeappendix{Kidd10insights} indeed reports to occur,
mainly in repetitive regions.
  \item Full sequencing of the discordant fosmid.  The full sequencing was
performed according to the protocols of the Human Genome Project, with error
rate pessimistically 1 in 10kb \citepappendix{IHGSC2004}.
  \item Breakpoint calling from the full sequence. Breakpoints were called via
human-guided inspection refined by optimal alignment \citepappendix{Kidd10insights}.
The main issue here is alignment ambiguity [Jeff Kidd, personal communication],
but this is not a concern to us, as we present SVs in in their full,
``ambiguous'' form, e.g., TAGCATTAG $\rightarrow$ TAG.
\end{enumerate}
We conclude that despite the complexity of their generation, the SV calls we
use are highly reliable.  Though we have not performed a full probabilistic
analysis, an error rate of $1\%$ seems conservative.

\subsection{Contaminated Datasets}
\label{ssec:corruption_details}
One common source of noise in short-read datasets are reads from a separate
genome, as might result from impurities introduced in sample preparation or
improperly sterilized sequencing machines. In order to benchmark algorithms against this type
of data, we created contaminated sets of short reads through a script that
walks the target and contaminate set of reads and outputs the contaminate read
at a specified likelihood and the target read otherwise. The contaminate reads
should have length and insert size as close as possible to the original reads;
ideally they should be from the same sequencer, both to simulate real-world
contaminated data and to prevent the aligner from filtering out obviously
unmatching reads. We chose $10\%$ as the likelihood of contamination for all
three of our contaminated datasets. The script used to create these datasets is
available to download.

\section{Recall and Precision}
\label{sec:rec_prec}
We now discuss in more detail our motivation for focusing on recall and
precision in the context of variant calling validation.  Consider the problem
of evaluating a variant caller against a ground truth validation set including
known variants and known locations lacking variation.
%only some of which are actually present in a particular validation dataset.
In this setting, each called variant is positive or negative (depending solely
on the caller) and true or false (depending also on the validation dataset), so
we can naively measure performance via four metrics -- true positives (TP), false positives (FP), true
negatives (TN) and false negatives (FN).  
%Now consider useful ways to summarize the behavior of the variant
%calling algorithm on the validation dataset.  
%Although one might be tempted to simply report all four metrics, 
However, it would be preferable to use a more succinct set of metrics by
distinguishing the behavior of the variant calling algorithm from intrinsic
properties of the validation dataset. 
%and in our particular
%genomics setting because $\TN$, e.g., ``all possible'' variant calls, is
%$exponentially large when dealing with structural variants.  

%To do so, we can reparametrize this four-dimensional space (see
%Section~\ref{proof:coord_change} in the Supplementary Material for further details) in terms
%of $\present = \TP + \FN$, $\absent = \FP + \TN$, $\precision = \TP / (\TP +
%\FP)$ and $\recall = \TP / (\TP + \FN)$.  

To find a convenient parametrization of this four dimensional space, first
note that the total numbers of variant presences ($\present = \TP + \FN$) and
absences ($\absent = \FP + \TN$) are intrinsic to the dataset, i.e.,
independent of the caller. Proposition~ \ref{prop:coordinate_change} uses these
definitions and provides a convenient reparametrization.

\begin{prop}
\label{prop:coordinate_change}
  No information is lost in the coordinate change
  \[ \{\TP, \TN, \FP, \FN\} \to
  \{\present, \absent, \recall, \precision\}\,.
  \]
\end{prop}
\begin{proof}
  The inverse transformation is readily verified to be
\begin{align*}
  \TP &= \recall \cdot \present \\
  \FN &= (1 - \recall) \cdot \present \\
  \FP &= \recall \cdot \present \cdot \left(\frac{1}{\precision} - 1\right) \\
  \TN &= \absent - \FP \qedhere \,.
\end{align*}
\end{proof}

Indeed, $\present$ and $\absent$ are intrinsic to the dataset, that is,
independent of the caller, and the space of caller behaviors on a particular
dataset can be expressed via the very widespread notions of recall and
precision.

\section{Noisy Validation Data Bounds}
\label{sec:noisy_validation_bounds}
In this section we present details of our analysis of noisy validation data.
We wish to derive bounds that will hold regardless of whether we are
considering SNPs, indels, or SVs, whether we are considering zygosity,
considering insertion sequence in addition to insertion length, etc.  To
abstract away these details, we assume that there is some set of possible
variants, at most one of which is actually present at each site in the
reference genome.  Then there are three kinds of validation error, at most one
of which may occur at a given site: a validated variant where no variant
actually exists, a validated variant at the site of a different actual variant,
and for a comprehensive dataset, the lack of a validated variant at a site
harboring an actual variant. 

\subsection{Simplified Binary Validation Error Setting}

Considering all three types of validation errors presents a fairly complex
combinatorial problem. In order to provide intuition for this problem, we first
consider the binary case, i.e., we assume that the set of possible variants at
each site has size 1.  This scenario is not directly applicable to variant
calling as it rules out the second kind of validation error described above,
but involves fairly straightforward arguments, and provides intuition for
subsequent results.  We present bounds on recall and precision for this setting
in Proposition~\ref{prop:binary_error}.

\begin{prop}
\label{prop:binary_error}
In the binary setting, the following bounds hold for recall and precision.
  
{Case 1: Positive Validation Only}:
\[
  \frac{\TP - E}{V - E} \leq \recall \leq \frac{\TP}{V - E} \,.
  \]
{Case 2: Positive and Negative Validation}:
  \begin{equation}
    \label{eq:precision_bounds}
  \frac{\TP - E}{P} \leq \precision \leq \frac{\TP + E}{P}
  \end{equation}
  %and upper and lower bounds for recall are
  % \[ \recall \lesseqgtr
  % \begin{cases}
  %   \frac{\TP}{V \mp E} &\mbox{if } \FN \leq \TP \pm E \\
  %   \frac{\TP \pm E}{V \pm E} &\mbox{otherwise.}
  % \end{cases}
  \[ \recall \leq
  \begin{cases}
    \frac{\TP}{V - E} &\mbox{if } \FN \leq \TP + E \\
    \frac{\TP + E}{V + E} &\mbox{otherwise,}
  \end{cases}
  \]
  \[ \recall \geq
  \begin{cases}
    \frac{\TP}{V + E} &\mbox{if } \FN \leq \TP - E \\
    \frac{\TP - E}{V - E} &\mbox{otherwise.}
  \end{cases}
  \]
  Note that if $\FN \leq \TP - E$, the recall bounds simplify to
  \[
  \frac{\TP}{V + E} \leq \recall \leq \frac{\TP}{V - E}
  \]
\end{prop}
\begin{proof}
\label{proof:binary_error}
  In the case of only positive validated labels, each validation error either
  goes $\TP \to \FP$ or $\FN \to \TN$, yielding the claimed bounds.

  For positive and negative validation data, there are four kinds of validation error:
  $\TP \leftrightarrow \FP$ and $\FN \leftrightarrow \TN$.  The bound on
  precision is immediate.  For recall, first consider finding the upper bound.
  The two kinds of validation error to consider here are $\FP \to \TP$ and $\FN
  \to \TN$.  Let $a$ be the number of validation errors of the first kind, and
  $b$ of the second.  Then
  \[
  \recall \leq \frac{\TP + a}{\TP + a + \FN - b}
  \]
  Eliminating the extraneous variable $b = E - a$, the derivative of this bound
  with respect to $a$ vanishes if and only if $\FN = \TP + E$.  Therefore, the
  bound is maximized at either $a = 0$ or $a = E$, corresponding to the two
  cases
  \[ \recall \leq
  \begin{cases}
    \frac{\TP}{V - E} &\mbox{if } \FN \leq \TP + E \\
    \frac{\TP + E}{V + E} &\mbox{otherwise.}
  \end{cases}
  \]
  Computing the lower bound for recall is completely analogous, with $-E$
  playing the role of $E$.
\end{proof}

%\subsection{Proof of Proposition~\ref{prop:real_error}}
\subsection{Proof of Proposition~1}
%Proposition~\ref{prop:real_error} considers the realistic non-binary setting.
Proposition~1 considers the realistic non-binary setting.
The resulting bounds look similar to those of
Proposition~\ref{prop:binary_error}, but require more involved analysis as we
deal with all three types of validation errors.  The complication of the
non-binary case is the notion of \emph{site}.  Each site harbors at most one
true variant, at most one call, and many true negatives, which we ignore
because they are invisible to the metrics of recall and precision.  There are
five kinds of sites:
  \begin{itemize}
  \item hits $\TP$ where the algorithm calls correctly;
  \item miscalls $M$, harboring $\FP$ and $\FN$, where there is a true variant
    but the algorithm calls a different variant;
  \item wrong silences $S$, harboring $\FN$, where the algorithm fails to make
    a call when it should;
  \item false alarms $L$, harboring $\FP$, where the algorithm calls but there
    is no true variant;
  \item typical sites $G$, with no call and no true variant.
  \end{itemize}
  Note that sites harboring true variants fall under $\TP$, $M$, or $S$; sites
  lacking true variants fall under $L$ or $G$, and are only relevant in the
  case of positive and negative validation data.

  There are three kinds of validation error:
  \begin{itemize}
  \item A validated variant where no variant actually exists.  This can effect
    $\TP \to L$, transferring $\TP \to \FP$; or $M \to L$ or $S \to G$,
    decreasing $\FN$.
  \item A validated variant at the site of a different actual variant.  This
    transfers $\TP \leftrightarrow M$, trading between $\TP$ and the
    combination of $\FP$ and $\FN$.
  \item For positive and negative validation data, the lack of a validated
  variant at a site harboring an actual variant.  This is merely the opposite
  of the first kind of validation error, either transferring $\FP \to \TP$ or
  increasing $\FN$.
  \end{itemize}

  We see immediately that error bounds for precision are the same as the binary
  case, as the only possible perturbations are $\TP \leftrightarrow \FP$; one
  direction also increases $\FN$ but this is not relevant to precision.

  Next consider the best-case upper bound of recall.  There are three ways
  validation error could increase recall:
  \begin{itemize}
  \item $M \to \TP$;
  \item $M \to L$ or $S \to G$;
  \item in the positive and negative validation case, $L \to \TP$.
  \end{itemize}
  Letting $a$, $b$, and $c$ be the number of validation errors of
  those three sorts, respectively, we have
  \[
  \recall \leq \frac{\TP + a + c}{V - b + c}
  \]
  In the best case, $c = 0$, so the upper bound will be the same for the
  positive and negative case as the positive-label case.  Since in the best case $a + b
  = E$, the upper bound is
  \[
  \recall \leq \frac{\TP + a}{V -E + a}
  \]
  with $0 \leq a \leq E$.  Since the derivative with respect to $a$ vanishes if
  and only if $E = \FN$, the maximum occurs at $a = E$ or $a = 0$,
  corresponding to the two claimed cases of the upper bound.

  Computing the worst-case lower bound of recall is similar, but shakes out a
  bit differently.  There are three ways validation error could decrease
  recall:
  \begin{itemize}
  \item $\TP \to M$;
  \item $\TP \to L$;
  \item in the positive and negative case, $L \to M$ or $G \to S$.
  \end{itemize}
  Letting $a$, $b$, and $c$ be the number of validation errors of those three
  sorts, respectively, we have
  \[
  \frac{\TP - a - b}{V + c - b} \leq \recall
  \]
  In the worst case, $b = 0$, and we obtain the desired lower bound in the case
  of only positive labels, where $c = 0$.  In the case of positive and negative 
  validation data, in the worst case $a + c = E$, yielding
  \[
  \frac{\TP - a}{V + E - a} \leq \recall,
  \]
  which gives us the desired lower bound.
%\end{proof}

\begin{figure}[h]
\renewcommand*\footnoterule{}
\centering
\begin{minipage}[t]{.5\textwidth}
\begin{algorithm}[H]
   \caption{\ourbm's Ambiguity Resolution Algorithm}
   \label{alg:ambiguity_algorithm}
\begin{algorithmic}
   \STATE {\bfseries Input:} Reference ($ref$), True VCF ($vcfTr$),
   Predicted VCF ($vcfPr$), Rescue window ($win$) \\
   \STATE {\bfseries Output:} Recall ($recall$), Precision ($prec$) \\
   \vspace{2mm}
	 \STATE $vcfPr$ = \textsc{CleanAndLeftNormalize}($vcfPr$) \\
	 \STATE $vcfTr$ = \textsc{CleanAndLeftNormalize}($vcfTr$) \\
	 \STATE $errors$ = \textsc{Compare}($vcfTr$, $vcfPr$) \comment{bleh} \\
   \STATE $rescued$ = \rescue($ref$, $vcfTr$, $vcfPr$, $errors$, $win$) \\
   \STATE $errors$ = $errors$ - $rescued$ \\
   \STATE $recall$, $prec$ = \textsc{ComputeResults}($vcfTr$, $errors$) \\
\end{algorithmic}
\end{algorithm}

\begin{algorithm}[H]
   \caption{\rescue}
   \label{alg:rescue}
\begin{algorithmic}
   \STATE {\bfseries Input:} Reference ($ref$), True VCF ($vcfTr$), Predicted
   VCF ($vcfPr$), Validation errors ($errors$), Rescue window ($win$) \\ 
   \STATE {\bfseries Output:} Set of rescued errors ($rescued$) \\ 
   \vspace{2mm}
   \STATE $rescued$ = $[\,]$ \\ 
   \textbf{for} $err$ \textbf{in} $errors$ 
   \STATE\hspace{2mm} $strTr$ = \textsc{ExpandAroundWindow}($ref$, $vcfTr$, $err$, $win$) \\
   \STATE\hspace{2mm} $strPr$ = \textsc{ExpandAroundWindow}($ref$, $vcfPr$, $err$, $win$) \\
   \STATE\hspace{2mm} \textbf{if} $strTr == strPr$ \\
   \STATE\hspace{5mm} $rescued$.append($error$) \\
\end{algorithmic}
\end{algorithm}

\comment{
\begin{algorithm}[H]
   \caption{\rescuega}
   \label{alg:rescue}
\begin{algorithmic}
   \STATE {\bfseries Input:} Reference ($ref$), True VCF ($vcfTr$), Predicted
   VCF ($vcfPr$), Validation errors ($errors$), Rescue window ($window$) \\ 
   \STATE {\bfseries Output:} Rescued errors ($rescued$) \\ 
   \vspace{2mm}
   \STATE $rescued$ = $\{\,\}$ \\ 
   \textbf{for} $error$ \textbf{in} $errors$ 
   \STATE\hspace{5mm} $seqHomTr, seqHetTr$ = \\
   \STATE\hspace{8mm} \textsc{ExpandErrorByGT}($ref$, $vcfTr$, $error$, $window$) \\
   \STATE\hspace{5mm} $seqHomPr, seqHetPr$ = \\
   \STATE\hspace{8mm} \textsc{ExpandErrorByGT}($ref$, $vcfPr$, $error$, $window$) \\
   \STATE\hspace{5mm} \textbf{if} $seqHomTr = seqHomPr$ \textbf{and} $seqHetTr = seqHetPr$ \\
   \STATE\hspace{10mm} $rescued$.add($error$) \\
\end{algorithmic}
\end{algorithm}
}

\end{minipage}
\end{figure}

\begin{table*}[ht!]
  \centering
  \small{
  \caption{Effect of ambiguity resolution on benchmarking GATK and mpileup on
  SNPs using the mouse dataset. The results illustrate the impact of each successive
  step of resolution, namely, cleaning, left normalization and rescuing.}
  \label{tab:ambiguity_mouse_extra}
  \begin{tabular}[r]{c|cc||cc}
  \hline
   & \multicolumn{2}{c||}{mpileup} & \multicolumn{2}{c}{GATK} \\
  \hline
   Strategy & Pre & Rec & Pre & Rec \\ 
  \hline
  Cleaning & 98.4 $\pm$ 0.3 & 87.2 $\pm$ 0.2 & 98.1 $\pm$ 0.2 & 93.5 $\pm$ 0.2 \\ 
 Normalization & 98.4 $\pm$ 0.3 & 87.2 $\pm$ 0.2 & 98.2 $\pm$ 0.2 & 93.5 $\pm$ 0.2 \\ 
 \rescue & 98.4 $\pm$ 0.3 & 87.3 $\pm$ 0.2 & 98.3 $\pm$ 0.2 & 94.7 $\pm$ 0.2 \\
%    \rescuega (window = $alt1$) & $49.0 \pm 0.0$ & $44.9 \pm 0.0$ & $48.5 \pm 0.0$ & $45.2 \pm 0.0$ \\ 
%    \rescuega (window = $alt2$) & $49.0 \pm 0.0$ & $44.9 \pm 0.0$ & $48.5 \pm 0.0$ & $45.2 \pm 0.0$ \\ 
    \multicolumn{5}{c}{}\\
\end{tabular}
}
\end{table*}

\begin{table*}[ht!]
  \centering
  \small{
  \caption{Effect of ambiguity resolution on benchmarking GATK, mpileup and Pindel on
  SNPs and indels using the Venter dataset. The results illustrate the impact of each successive
  step of resolution, namely, cleaning, left normalization and rescuing. Error bounds are
  excluded since there is no uncertainty in the Venter validation data. Top: SNPs,
  Bottom: Indels.}
  \label{tab:ambiguity_venter}
  \begin{tabular}[r]{c|cc||cc}
  \hline
   & \multicolumn{2}{c||}{mpileup} & \multicolumn{2}{c}{GATK} \\
  \hline
   Strategy & Pre & Rec & Pre & Rec \\ 
  \hline
 Cleaning & 96.9 & 97.0 & 97.4 & 91.5 \\ 
 Normalization & 96.9 & 97.0 & 97.4 & 91.5 \\ 
 \rescue & 98.7 & 97.0 & 99.3 & 91.7 \\
%    \rescuega (window = $alt1$) & $49.0 \pm 0.0$ & $44.9 \pm 0.0$ & $48.5 \pm 0.0$ & $45.2 \pm 0.0$ \\ 
%    \rescuega (window = $alt2$) & $49.0 \pm 0.0$ & $44.9 \pm 0.0$ & $48.5 \pm 0.0$ & $45.2 \pm 0.0$ \\ 
    \multicolumn{5}{c}{}\\
  \end{tabular}
  \begin{tabular}[r]{c|cc|cc||cc|cc||cc|cc}
  \hline
   & \multicolumn{4}{c||}{mpileup} & \multicolumn{4}{c||}{GATK} & \multicolumn{4}{c}{Pindel} \\
   & \multicolumn{2}{c}{Insertions} & \multicolumn{2}{c||}{Deletions} & \multicolumn{2}{c}{Insertions} & \multicolumn{2}{c||}{Deletions} & \multicolumn{2}{c}{Insertions} & \multicolumn{2}{c}{Deletions}\\
  \hline
   Strategy & Pre & Rec & Pre & Rec & Pre & Rec & Pre & Rec & Pre & Rec  & Pre & Rec\\  
  \hline
  Cleaning & 73.2 & 5.1 & 78.9 & 6.5 & 14.7 & 13.3 & 15.8 & 14.3 & 10.6 & 7.5 & 12.7 & 9.7 \\ 
 Normalization & 84.5 & 69.7 & 90.7 & 74.2 & 87.5 & 86.3 & 89.8 & 89.2 & 92.3 & 71.2 & 93.8 & 78.3 \\ 
 \rescue & 85.8 & 70.9 & 91.3 & 75.2 & 90.6 & 88.1 & 92.3 & 90.9 & 92.7 & 72.1 & 94.1 & 79.2 \\
%    \rescuega (window = $alt1$) & $-$ & $-$ & $-$ & $-$ & $74.3 \pm 0.4$ & $94.3 \pm 0.5$ & $66.5 \pm 0.3$ & $96.3 \pm 0.5$ \\
%    \rescuega (window = $alt2$) & $-$ & $-$ & $-$ & $-$ & $74.3 \pm 0.4$ & $94.3 \pm 0.5$ & $66.5 \pm 0.3$ & $96.3 \pm 0.5$ \\
  \end{tabular}
}
\end{table*}

\section{Ambiguity Resolution Details}
\label{sec:ambiguity_details}
We now describe the details of the ambiguity resolution algorithm we use when
benchmarking variant callers. Input VCF files are first cleaned (removing
homozygous reference calls, and calls where the reference and alternate alleles
matched) and left-shifted. Next, we perform an initial evaluation in which we
loop through the called variants in the predicted and ground truth VCFs.
SNPs and short insertions and deletions are strictly compared. Specifically, 
\begin{itemize}
\item calls where position and alleles match identically are marked as true
positives;
\item calls in the ground truth VCF but not the predicted VCF, either due to
mismatching alleles or no corresponding variant in the predicted VCF, are
marked as false negatives; and 
\item variants in the predicted VCF that have no
matching call in the truth VCF are marked as false positives. 
\end{itemize}
Because there are several ways to represent a mutation or set of mutations, we
subsequently run the \rescue algorithm in an attempt to match false negatives with false
positives within a small window.  Algorithm~\ref{alg:ambiguity_algorithm}
summarizes \ourbm's ambiguity resolution procedure while
Algorithm~\ref{alg:rescue} outlines the steps of the \rescue procedure. 

Due to the symmetry of the problem (an allele which is equivalent in the two
sets will generate both false positives and false negatives), we need only
examine one of these sets. Because it is possible to represent the reference
allele itself by a combination of variants which cancel out (for instance an
insertion followed by a deletion), we penalize this potential behavior by
attempting to rescue only putative
false negatives. Tables~\ref{tab:ambiguity_mouse_extra} and
\ref{tab:ambiguity_venter} provide additional evidence (supporting
%Table~\ref{tab:ambiguity_mouse} in the main text) of the impact of these three
Table~3 in the main text) of the impact of these three
ambiguity resolution steps on variant calling accuracy.

We take the following approach in \rescue for dealing with overlapping alleles
within each set (i.e. several overlapping FPs). First, we exclude alleles that
have the same starting position; when presented with such a site we select the
first of these alleles. Second, for alleles that overlap,
%such as a SNP within a deletion, 
all possible combination of variants within the window are checked
for equivalence, up to a limit of 16. For instance, consider a 100bp window
  containing both two overlapping deletions (\texttt{var1} and \texttt{var2})
  and a SNP that overlaps a deletion (\texttt{var3} and \texttt{var4}). In
  order to test if these variants result in the same nucleotide sequence as
  those in the truth VCF, we must construct a reference string from them;
  but the overlaps render this impossible. Instead, in this example we generate
  four sequences: those resulting from the pairs (\texttt{var1},\texttt{var3}),
  (\texttt{var2},\texttt{var3}), (\texttt{var1},\texttt{var4}), and
  (\texttt{var2},\texttt{var4}).  At a window, each VCF generates a set of
  sequences in this way, and all pairs of nucleotide sequences are checked for
  identity. The first pair of sequences which is identical between the two VCFs
  is marked as a match. The false negatives giving rise to the true sequence
  become true positives, and the false positives giving rise to the comparison
  sequence are removed from accounting.

\begin{table}[ht]
  \caption{Computational performance for BWA alignment.}
  \centering
  \label{tab:align}
  \begin{tabular}[r]{c|rr}
  dataset & \multicolumn{2}{c}{BWA} \\
  \hline
    & Hours & Cost \\
  \hline
  Venter & $60h$ & $\$150$ \\
  contam. Venter & $23h$ & $\$57$ \\
  NA12878 & $52h$ & $\$130$ \\
  contam. NA12878 & $45h$ & $\$112$ \\
  NA18507 & $37h$ & $\$92$ \\
  NA19240 & $36h$ & $\$90$ \\
  mouse & $106h$ & $\$265$ \\
  contam. mouse & $77h$ & $\$192$
  \end{tabular}
  \end{table}

\section{Alignment Statistics}\label{sec:align}

We performed alignment on all benchmarking datasets to evaluate different
variant callers in Section~\ref{section:evaluation}.  All our alignments are
%variant callers in Section~3.  All our alignments are
performed using BWA 6.1 \citepappendix{LiDurbin09} with default settings,
except with trimming parameter 15.  Table~\ref{tab:align} presents the time and
cost to run BWA on AWS for each dataset.  The resulting BAM files are available
at \ourwebsite.  

We now provide a summary of the alignment results, starting with
Table~\ref{tab:align-stats}, which presents alignment results for the eight
datasets we consider here.  We first restrict attention to reads aligning
uniquely,\footnote{We define a read as ``aligning uniquely'' if it has a single
best alignment, as opposed to zero or multiple alignments.  We implement this
definition by excluding SAM flags ``Usfd'' (meaning unmapped, not primary, QC
failure, and optical or PCR duplicate) and ``XT'' tag values ``Repeat'' and
``N'' (meaning not mapped).} yielding the first column of
Table~\ref{tab:align-stats}.  The final three columns indicate the percentage
of reads that align uniquely with edit distance at most 5, 2, or 0. 
% Table~\ref{tab:align-stats} is only meant as a rough guide; we do
%not necessarily advocate the use of edit distance thresholds for benchmarking
%alignment.
We also present alternative alignment statistics in
Table~\ref{tab:quality-align-stats} by restricting attention to reads whose
quality scores, after trimming, average 25 and have no more than two ``!''
values.  Note that the mouse (Section~\ref{subsec:mouse}) reads have the worst
%values.  Note that the mouse (Section~2.1.2) reads have the worst
alignment statistics, but the statistics improve dramatically after we perform
quality filtering and presumably each variant calling algorithm performs a
similar type of filtering before processing the reads. 

Note that the percentage of reads aligned for Venter and NA12878 and their
corresponding contaminated datasets are almost identical. However, fewer of the
contaminated mouse dataset reads were aligned than in the mouse dataset; since
the contaminate reads were from a different species, the aligner was evidently
better able to filter many of them out.

\begin{table}[ht!]
 \centering
 \caption{Alignment statistics for BWA.}\label{tab:align-stats}
 \begin{tabular}[c]{l|r|r|r|r}
   dataset & aligned & $\leq 5$ & $\leq 2$ & $= 0$ \\
   \hline
   Venter & 88.0 & 99.6 & 94.5 & 60.4 \\
   contam. Venter & 87.8 & 99.7 & 94.7 & 60.8 \\
   NA12878 & 91.3 & 99.6 & 97.7 & 83.8 \\
   contam. NA12878 & 90.9 & 99.6 & 97.8 & 80.9 \\
   NA18507 & 94.4 & 99.6 & 96.9 & 52.3 \\
   NA19240 & 86.2 & 99.5 & 95.7 & 68.8 \\
   Mouse & 76.2 & 99.6 & 93.1 & 66.8 \\
   %B6 mouse, long-insert reads & 47.2 & 46.5 & 41.3 & 27.0 \\
   contam. Mouse & 68.4 & 99.6 & 93.1 & 66.7
 \end{tabular}
\end{table}

\begin{table}[ht!]
 \centering\caption{Quality-filtered alignment statistics for BWA.}\label{tab:quality-align-stats}
 \begin{tabular}[c]{l|r|r|r|r}
   dataset & aligned & $\leq 5$ & $\leq 2$ & $= 0$ \\
   \hline
   Venter & 95.4 & 99.8 & 95.5 & 61.4 \\
   contam. Venter & 95.4 & 99.8 & 95.5 & 61.4 \\
   NA12878 & 94.8 & 99.6 & 97.9 & 84.5 \\
   contam. NA12878 & 94.7 & 99.7 & 98.0 & 81.5 \\
   NA18507 & 94.4 & 99.6 & 97.0 & 52.9 \\
   NA19240 & 94.0 & 99.6 & 96.8 & 71.9 \\
   Mouse & 86.8 & 99.6 & 93.4 & 68.6 \\
   contam. Mouse & 83.2 & 99.6 & 93.4 & 68.5

 \end{tabular}
\end{table}

\section{Experimental Details}
\label{sec:experimental_details}

In this section we describe the details of the experimental setup for the
results presented in Section~\ref{section:evaluation}.
%results presented in Section~3.

\textbf{Computing Platform}: All experiments were run on Amazon AWS using
cc2.8xlarge instances (8 quad-core processors, 60.5 GB of RAM).  Genomics data
was stored on EBS and local (ephemeral) storage was configured in RAID0.  

\textbf{SAMtools mpileup}:
We used SAMtools version 0.1.19-44428cd, which we obtained from the SAMtools
sourceforge download
page.\footnote{\url{http://sourceforge.net/projects/samtools}}  We followed the
`basic Command line' pipeline described on the official mpileup
website,\footnote{\url{http://samtools.sourceforge.net/mpileup.shtml}} using
the mpileup option `-C50' as per the recommendation on the mpileup website.
We ran mpileup in parallel by chromosome, and concatenated the resulting VCFs
using the `vcf-concat'
command\footnote{\url{http://vcftools.sourceforge.net/perl_module.html#vcf-concat}}
from VCFtools (note that cc2.8xlarge Amazon instances support up to 32
threads, which is more than the number of chromosomes).

\textbf{GATK}: We ran GATK version 2.6, using the parameters and the pipeline
described by the GATK best practices
wiki.\footnote{\url{http://gatkforums.broadinstitute.org/discussion/15/retired-best-practice-variant}
\url{-detection-with-the-gatk-v3}}
The pipeline we executed used GATK's Queue interface to GridEngine, thus
ensuring that GATK used as many cores as possible (actual core usage varied
at different phases of the analysis due to the ways in which the various steps
of the pipeline parallelized).  
%Additionally, we note that GATK
%version 2.7 was released during our benchmarking experiments and was thus not
%used for our experiments.  Nonetheless, for the purpose of illustrating the
%utility of \ourbm, GATK version 2.6 indeed suffices. 

The final stage in the GATK best practices pipeline is to use known variants as
training data to establish the probability of each individual call's accuracy,
so that low-probability calls can be filtered out. For the synthetic and
sampled human datasets, we used the hg19 resources from the GATK resource
bundle.\footnote{\url{http://www.broadinstitute.org/gatk/guide/article?id=1213}} For
the mouse datasets, the dbSNP variants did not overlap enough with our
reference to provide a statistically significant prior, so we did not include
the variant recalibration step in our GATK mouse pipeline.

\textbf{Pindel}: We ran Pindel version 0.2.4w (May 31 2013) which we obtained
from the Pindel
repository.\footnote{\url{https://github.com/genome/pindel/tree/3790e78b969dca678896e6e98ec09b703d574f13}}
We directed Pindel to run on all chromosomes but otherwise used the default
parameters. Pindel appears to  become memory-bound when working on a
high-coverage sample.  In particular, on the NA12878 sample, we were unable to
run it with multiple threads on the cc2.8xlarge instances (which have only 60.5
GB of memory); however, on a separate machine with 128GB of memory, we were in
fact able to run it with multiple threads.  Hence, we did not run Pindel in
parallel for any of the datasets. Pindel ran for over 400 hours on the
contaminated Venter, contaminated NA12878 and mouse datasets; we terminated these
jobs before they completed.

\textbf{BreakDancer}:
We downloaded breakdancer-max version 1.2.6 (commit 83efb8e) from the
BreakDancer software
repository,\footnote{\url{http://gmt.genome.wustl.edu/breakdancer/1.2/install.html}}
and ran it with its default parameters.  In order to be consistent with our
experiments with Pindel, we chose not parallelize BreakDancer (though it is
indeed possible to do so by  scheduling distinct jobs for each chromosome,
along with a separate job to identify interchromosomal translocations).  
BreakDancer's output provides, for each SV, left and right breakpoints and a
length; we use the left breakpoint and length to convert to VCF. BreakDancer ran for over 
400 hours on the contaminated NA12878 dataset and was terminated before completion.

\textbf{Evaluation Script}:
We first preprocess the input VCF files (predicted and ground truth) via
cleaning and left normalization and then compute the counts of true positives,
false positives (where applicable), and false negatives using the ambiguity
resolution algorithms described in 
%Section~2.2.2 and Section~\ref{sec:ambiguity_details}.  
%Given these counts, we use Proposition~1
Section~\ref{ssec:ambiguity} and Section~\ref{sec:ambiguity_details}.  
Given these counts, we use Proposition~\ref{prop:real_error} 
to obtain error bounds on recall and
precision. Moreover, while evaluating SVs, if an SV has multiple alternate
alleles, i.e., multiple possible lengths, we score it as correct if any
alternate allele yields the correct length within the error tolerance for
length.  Within each of the two categories of SVs, insertions and deletions, we
score each validated variant against the closest predicted variant. 

\textbf{Known false positives}:
In our sampled human SNP data, in addition to sites known to be polymorphic in
the sample, there are alleles which through validation are known to be false
positives, due to sequencing or alignment artifacts rather than underlying
genomic variation.  Counts of these sites are tabulated to compute estimates
of precision for SNP calling (though not unbiased ones). 
%of false-positive rates in sequenced (rather than
%simulated) data between two callers.

\section{Additional Experiments}
We now present additional benchmarking results.  Tables~\ref{tab:indel_del_all} and
\ref{tab:indel_ins} show results on small insertions and deletions, while
Tables~\ref{tab:sv_ins} and \ref{tab:sv_ins_0bp} provide results for long insertions. It should be noted that
the long insertions results are fairly poor across the board; for the mouse and 
contaminated mouse results, we see very large margins of error because both callers
predicted very few long insertions relative to the number of erroneous variants
we expect to find in our validated truth data given the upper-bounded rate of
error.  Moreover, Table~\ref{tab:perf_stats} presents additional computational
%performance results to supplement the results in Table~8 in the main text.
performance results to supplement the results in Table~\ref{tab:venter_stats} in the main text.

\bibliographystyleappendix{plainnat}
\bibliographyappendix{benchmarking_refs}
\clearpage
\begin{landscape}
\begin{table}
  \centering
  \caption{Benchmarking results for small deletion (including Pindel results).}
  \label{tab:indel_del_all}
  \begin{tabular}[r]{r|rrrr||rrrr||rrrr}
   dataset & \multicolumn{4}{c||}{mpileup} & \multicolumn{4}{c||}{GATK} & \multicolumn{4}{c}{Pindel} \\
  \hline
     & \multicolumn{1}{c}{Hours} & \multicolumn{1}{c}{Cost} & \multicolumn{1}{c}{Pre} & \multicolumn{1}{c||}{Rec} &\multicolumn{1}{c}{Hours} & \multicolumn{1}{c}{Cost} & \multicolumn{1}{c}{Pre} & \multicolumn{1}{c||}{Rec} & \multicolumn{1}{c}{Hours} & \multicolumn{1}{c}{Cost} & \multicolumn{1}{c}{Pre} & \multicolumn{1}{c}{Rec}\\
  \hline
Venter & $2 h$ & $\$ 5$ & $91.3\% \pm 0.0$ &  $75.2\% \pm 0.0$ & $57 h$ & $\$ 142$ & $92.3\% \pm 0.0$ & $90.9\% \pm 0.0$& $80 h$ & $\$ 200$ &$94.0\% \pm 0.0$ & $79.8\% \pm 0.0$ \\
contam. Venter & $3 h$ & $\$ 8$ & $91.7\% \pm 0.0$ &  $71.7\% \pm 0.0$ & $75 h$ & $\$188$ & $92.4\% \pm 0.0$ & $90.5\% \pm 0.0$ & \tooslow & \toomuch & - & -  \\
mouse & $6 h$ & $\$15$ & $79.0\% \pm 0.4$ &  $85.9\% \pm 0.4$ & $107 h$ & $\$268$ & $81.5\% \pm 0.4$ & $95.8\% \pm 0.4$ & \tooslow & \toomuch & - & -  \\
contam. mouse & $5 h$ & $\$13$ & $80.4\% \pm 0.4$ &  $84.9\% \pm 0.4$ & $96 h$ & $\$240$ & $82.8\% \pm 0.4$ & $95.6\% \pm 0.4$& $345 h$ & $\$ 862$ &$85.9\% \pm 0.4$ & $81.3\% \pm 0.4$ \\ \\\end{tabular}
\end{table}

\begin{table}
  \centering
  \caption{Benchmarking results for small insertions.}
  \label{tab:indel_ins}
  \begin{tabular}[r]{r|rrrr||rrrr||rrrr}
   dataset & \multicolumn{4}{c||}{mpileup} & \multicolumn{4}{c||}{GATK} & \multicolumn{4}{c}{Pindel} \\
  \hline
     & \multicolumn{1}{c}{Hours} & \multicolumn{1}{c}{Cost} & \multicolumn{1}{c}{Pre} & \multicolumn{1}{c||}{Rec} &\multicolumn{1}{c}{Hours} & \multicolumn{1}{c}{Cost} & \multicolumn{1}{c}{Pre} & \multicolumn{1}{c||}{Rec} & \multicolumn{1}{c}{Hours} & \multicolumn{1}{c}{Cost} & \multicolumn{1}{c}{Pre} & \multicolumn{1}{c}{Rec}\\
  \hline
Venter & $2 h$ & $\$ 5$ & $85.8\% \pm 0.0$ &  $70.9\% \pm 0.0$ & $57 h$ & $\$ 142$ & $90.6\% \pm 0.0$ & $88.1\% \pm 0.0$& $80 h$ & $\$ 200$ &$92.5\% \pm 0.0$ & $72.9\% \pm 0.0$ \\
contam. Venter & $3 h$ & $\$ 8$ & $87.8\% \pm 0.0$ &  $68.7\% \pm 0.0$ & $75 h$ & $\$188$ & $90.9\% \pm 0.0$ & $87.5\% \pm 0.0$ & \tooslow & \toomuch & - & -  \\
mouse & $6 h$ & $\$15$ & $87.8\% \pm 0.4$ &  $76.6\% \pm 0.4$ & $107 h$ & $\$268$ & $91.0\% \pm 0.4$ & $91.5\% \pm 0.4$ & \tooslow & \toomuch & - & -  \\
contam. mouse & $5 h$ & $\$13$ & $88.6\% \pm 0.5$ &  $75.8\% \pm 0.4$ & $96 h$ & $\$240$ & $91.1\% \pm 0.4$ & $90.9\% \pm 0.4$& $345 h$ & $\$ 862$ &$92.4\% \pm 0.5$ & $71.7\% \pm 0.4$ \\\end{tabular}
\end{table}

\begin{table}
  \centering
  \caption{Benchmarking results for long insertions (approximate evaluation with breakpoint tolerance of 100bp). }
  \label{tab:sv_ins}
  \begin{tabular}[r]{r|rrrr||rrrr}
   dataset & \multicolumn{4}{c||}{Pindel} & \multicolumn{4}{c}{BreakDancer} \\
  \hline
     & \multicolumn{1}{c}{Hours} & \multicolumn{1}{c}{Cost} & \multicolumn{1}{c}{Pre} & \multicolumn{1}{c||}{Rec} & \multicolumn{1}{c}{Hours} & \multicolumn{1}{c}{Cost} & \multicolumn{1}{c}{Pre} & \multicolumn{1}{c}{Rec}\\
  \hline
Venter & $80 h$ & $\$ 200$ & $24.8\% \pm 0.0$ & $10.2\% \pm 0.0$& $2 h$ & $\$ 5$ & $4.2\% \pm 0.0$ & $0.0\% \pm 0.0$ \\
contam. Venter & \tooslow & \toomuch & - & - & $1 h$ & $\$ 2$ & $0.0\% \pm 0.0$ & $0.0\% \pm 0.0$ \\
NA12878 & $81 h$ & $\$203$ &  -  & $0.0\% \pm 3.0$& $3 h$ & $\$7$ &  -  & $0.0\% \pm 3.0$ \\
contam. NA12878 & \tooslow & \toomuch & - & - & \tooslow & \toomuch & - & -  \\
NA18507 & $168 h$ & $\$420$ &  -  & $0.0\% \pm 5.0$& $3 h$ & $\$7$ &  -  & $0.0\% \pm 5.0$ \\
NA19240 & $221 h$ & $\$ 552$ &  -  & $0.0\% \pm 2.5$& $29 h$ & $\$73$ &  -  & $0.0\% \pm 2.5$ \\
mouse & \tooslow & \toomuch & - & - & $4 h$ & $\$ 10$ & $1.8\% \pm 35.7$ & $0.0\% \pm 0.5$ \\
contam. mouse & $345 h$ & $\$ 862$ & $8.2\% \pm 82.7$ & $1.3\% \pm 0.5$& $3 h$ & $\$7$ & $0.0\% \pm 0.0$ & $0.0\% \pm 0.5$ \\ 
\end{tabular}
\end{table}

\begin{table}
  \caption{Benchmarking results for long insertions (exact evaluation). }
    \centering
  \label{tab:sv_ins_0bp}
  \begin{tabular}[r]{r|rrrr||rrrr}
   dataset & \multicolumn{4}{c||}{Pindel} & \multicolumn{4}{c}{Breakdancer} \\
  \hline
     & \multicolumn{1}{c}{Hours} & \multicolumn{1}{c}{Cost} & \multicolumn{1}{c}{Pre} & \multicolumn{1}{c||}{Rec} & \multicolumn{1}{c}{Hours} & \multicolumn{1}{c}{Cost} & \multicolumn{1}{c}{Pre} & \multicolumn{1}{c}{Rec}\\
  \hline
Venter & $80 h$ & $\$ 200$ & $0.0\% \pm 0.0$ & $9.7\% \pm 0.0$& $2 h$ & $\$ 5$ & $0.0\% \pm 0.0$ & $0.0\% \pm 0.0$ \\
contaminated Venter & \tooslow & \toomuch & - & - & $1 h$ & $\$ 2$ & $0.0\% \pm 0.0$ & $0.0\% \pm 0.0$ \\
NA12878 & $81 h$ & $\$203$ &  -  & $0.0\% \pm 3.0$& $3 h$ & $\$7$ &  -  & $0.0\% \pm 3.0$ \\
contaminated NA12878 & \tooslow & \toomuch & - & - & \tooslow & \toomuch & - & -  \\
NA18507 & $168 h$ & $\$420$ &  -  & $0.0\% \pm 5.0$& $3 h$ & $\$7$ &  -  & $0.0\% \pm 5.0$ \\
NA19240 & $221 h$ & $\$ 552$ &  -  & $0.0\% \pm 2.5$& $29 h$ & $\$73$ &  -  & $0.0\% \pm 2.5$ \\
mouse & \tooslow & \toomuch & - & - & $4 h$ & $\$ 10$ & $0.0\% \pm 35.7$ & $0.0\% \pm 0.5$ \\
contaminated mouse & $345 h$ & $\$ 862$ & $0.0\% \pm 82.7$ & $1.3\% \pm 0.5$& $3 h$ & $\$7$ & $0.0\% \pm 0.0$ & $0.0\% \pm 0.5$ \\\end{tabular}
\end{table}
\end{landscape}

\begin{table*}
  \caption{Detailed performance statistics for all datasets and algorithms.}
  \centering
  \label{tab:perf_stats}
  \begin{tabular}[r]{r|r||rrrrrrr}
  dataset & caller & Clock Time & Cost & CPU Time & Max Threads & Max disk & Max memory & Avg. memory \\
  \hline
Venter & mpileup & $2 h$ & $\$5$ & $28 h$ & $24$ & $0.9$ GB & $59.1$ GB & $58.8$ GB \\
Venter & GATK & $57 h$ & $\$143$ & $179 h$ & $16$ & $530$ GB & $59.1$ GB & $53.4$ GB \\
Venter & Pindel & $80 h$ & $\$200$ & $80 h$ & $1$ & $3.4$ GB & $59.1$ GB & $58.7$ GB \\
Venter & BreakDancer & $2 h$ & $\$5$ & $2 h $ & $1$ & $0.02$ GB & $59.1$ GB & $58.7$ GB \\
\hline
contam. Venter & mpileup & $3 h$ & $\$8$ & $28$ & $24$ & $1.1$ GB & $59.1$ GB & $57.5$ GB \\
contam. Venter & GATK & $75 h$ & $\$188$ & $152 h$ & $16$ & $536$ GB & $59.1$ GB & $52.8$ GB \\
contam. Venter & Pindel & - & - & - & - & - & - & - \\
\hline
NA12878 & mpileup & $5 h$ & $\$13$ & $52 h$ & $24$ & $1.2$ GB & $59.1$ GB & $56.2$ GB \\
NA12878 & GATK & $86 h$ & $\$215$ & $174 h$ & $16$ & $852$ GB & $59.1$ GB & $55.0$ GB \\
NA12878 & Pindel & $81 h$ & $\$203$ & $78 h$ & $1$ & $4.4$ GB & $59.1$ GB & $58.4$ GB\\
NA12878 & BreakDancer & $3 h$ & $\$7$ & $3 h$ & $1$ & $0.05$ GB & $59.1$ GB & $58.6$ GB \\
\hline
contam. NA12878 & mpileup & $5 h$ & $\$13$ & $54 h$ & $24$ & $1.6$ GB & $59.1$ GB & $57.1$ GB\\
contam. NA12878 & GATK & $110 h$ & $\$275$ & $242 h$ & $16$ & $878$ GB & $59.1$ GB & $56.6$ GB\\
contam. NA12878 & Pindel & - & - & - & - & - & - & - \\
contam. NA12878 & BreakDancer & - & - & - & - & - & - & - \\
\hline
NA18507 & mpileup & $4 h$ & $\$10$ & $43 h$ & $24$ & $1.8$ GB & $59.1$ GB & $57.4$ GB \\
NA18507 & GATK & $154 h$ & $\$385$ & $234 h$ & $16$ & $710$ GB & $59.1$ GB & $41.9$ GB \\
NA18507 & Pindel & $168 h$ & $\$420$ & $167 h$ & $1$ & $4.4$ GB & $59.1$ GB & $58.8$ GB\\
NA18507 & BreakDancer & $3 h$ & $\$7$ & $3 h$ & $1$ & $0.08$ GB & $59.1$ GB & $58.4$ GB\\
\hline
NA19240 & mpileup & $4 h$ & $\$10$ & $46 h$ & $24$ & $1.4$ GB & $59.1$ GB & $57.0$ GB\\
NA19240 & GATK & $167 h$ & $\$418$ & $317 h$ & $16$ & $335$ GB & $59.1$ GB & $50.9$ GB \\
NA19240 & BreakDancer & $29 h$ & $\$73$ & $29 h$ & $1$ & $0.09$ GB & $59.1$ GB & $58.9$ GB \\
\hline
mouse & mpileup & $6 h$ & $\$15$ & $71 h$ & $24$ & $1.4$ GB & $59.1$ GB & $59.0$ GB \\
mouse & GATK & $107 h$ & $\$268$ & $434 h$ & $16$ & $839$ GB & $59.1$ GB & $56.9$ GB\\
mouse & Pindel & - & - & - & - & - & - & - \\
mouse & BreakDancer & $4 h$ & $\$10$ & $4 h$ & $1$ & $0.1$ GB & $59.1$ GB & $56.7$ GB \\
\hline
contam. mouse & mpileup & $5 h$ & $\$13$ & $60 h$ & $24$ & $0.7$ GB & $59.1$ GB & $56.2$ GB \\
contam. mouse & GATK & $96 h$ & $\$240$ & $351 h$ & $16$ & $786$ GB & $59.1$ GB & $56.1$ GB\\
contam. mouse & BreakDancer & $3 h$ & $\$7$ & $3 h$ & $1$ & $0.08$ GB & $59.1$ GB & $56.5$ GB \\
  \end{tabular}
\end{table*}

\end{document}